\PassOptionsToPackage{unicode}{hyperref}
\PassOptionsToPackage{hyphens}{url}
\PassOptionsToPackage{dvipsnames,svgnames,x11names}{xcolor}
\documentclass[
  12pt]{article}

\usepackage{amsmath,amssymb,bm,bbm,amsthm}
\usepackage{iftex}
\usepackage{enumitem}
\usepackage{natbib}
\usepackage{algorithmicx, algorithm, setspace,algpseudocode}
\ifPDFTeX
  \usepackage[T1]{fontenc}
  \usepackage[utf8]{inputenc}
  \usepackage{textcomp} 
\else 
  \usepackage{unicode-math}
  \defaultfontfeatures{Scale=MatchLowercase}
  \defaultfontfeatures[\rmfamily]{Ligatures=TeX,Scale=1}
\fi
\usepackage{lmodern}
\ifPDFTeX\else  
\fi
\IfFileExists{upquote.sty}{\usepackage{upquote}}{}
\IfFileExists{microtype.sty}{
  \usepackage[]{microtype}
  \UseMicrotypeSet[protrusion]{basicmath} 
}{}
\makeatletter
\@ifundefined{KOMAClassName}{
  \IfFileExists{parskip.sty}{%
    \usepackage{parskip}
  }{
    \setlength{\parindent}{0pt}
    \setlength{\parskip}{6pt plus 2pt minus 1pt}}
}{
  \KOMAoptions{parskip=half}}
\makeatother
\usepackage{xcolor}
\makeatletter
\ifx\paragraph\undefined\else
  \let\oldparagraph\paragraph
  \renewcommand{\paragraph}{
    \@ifstar
      \xxxParagraphStar
      \xxxParagraphNoStar
  }
  \newcommand{\xxxParagraphStar}[1]{\oldparagraph*{#1}\mbox{}}
  \newcommand{\xxxParagraphNoStar}[1]{\oldparagraph{#1}\mbox{}}
\fi
\ifx\subparagraph\undefined\else
  \let\oldsubparagraph\subparagraph
  \renewcommand{\subparagraph}{
    \@ifstar
      \xxxSubParagraphStar
      \xxxSubParagraphNoStar
  }
  \newcommand{\xxxSubParagraphStar}[1]{\oldsubparagraph*{#1}\mbox{}}
  \newcommand{\xxxSubParagraphNoStar}[1]{\oldsubparagraph{#1}\mbox{}}
\fi
\makeatother

\usepackage{longtable,booktabs,array}
\usepackage{calc} 
\usepackage{etoolbox}
\makeatletter
\patchcmd\longtable{\par}{\if@noskipsec\mbox{}\fi\par}{}{}
\makeatother
\IfFileExists{footnotehyper.sty}{\usepackage{footnotehyper}}{\usepackage{footnote}}
\makesavenoteenv{longtable}
\usepackage{graphicx}
\makeatletter
\def\maxwidth{\ifdim\Gin@nat@width>\linewidth\linewidth\else\Gin@nat@width\fi}
\def\maxheight{\ifdim\Gin@nat@height>\textheight\textheight\else\Gin@nat@height\fi}
\makeatother
\setkeys{Gin}{width=\maxwidth,height=\maxheight,keepaspectratio}
\makeatletter
\def\fps@figure{htbp}
\makeatother

\makeatletter
\@ifpackageloaded{caption}{}{\usepackage{caption}}
\AtBeginDocument{%
\ifdefined\contentsname
  \renewcommand*\contentsname{Table of contents}
\else
  \newcommand\contentsname{Table of contents}
\fi
\ifdefined\listfigurename
  \renewcommand*\listfigurename{List of Figures}
\else
  \newcommand\listfigurename{List of Figures}
\fi
\ifdefined\listtablename
  \renewcommand*\listtablename{List of Tables}
\else
  \newcommand\listtablename{List of Tables}
\fi
\ifdefined\figurename
  \renewcommand*\figurename{Figure}
\else
  \newcommand\figurename{Figure}
\fi
\ifdefined\tablename
  \renewcommand*\tablename{Table}
\else
  \newcommand\tablename{Table}
\fi
}
\@ifpackageloaded{float}{}{\usepackage{float}}
\floatstyle{ruled}
\@ifundefined{c@chapter}{\newfloat{codelisting}{h}{lop}}{\newfloat{codelisting}{h}{lop}[chapter]}
\floatname{codelisting}{Listing}

\makeatother
\makeatletter
\@ifpackageloaded{caption}{}{\usepackage{caption}}
\@ifpackageloaded{subcaption}{}{\usepackage{subcaption}}
\makeatother

\ifLuaTeX
  \usepackage{selnolig}  
\fi
\usepackage[]{natbib}
\usepackage{bookmark}
\IfFileExists{xurl.sty}{\usepackage{xurl}}{} 
\hypersetup{
  pdftitle={Title},
  pdfauthor={Author 1; Author 2},
  pdfkeywords={3 to 6 keywords, that do not appear in the title},
  colorlinks=true,
  linkcolor={Red},
  filecolor={Maroon},
  citecolor={Blue},
  urlcolor={Blue},
  pdfcreator={LaTeX via pandoc}}

\newtheorem{theorem}{Theorem}

\newtheorem{lemma}{Lemma}
\newtheorem{proposition}{Proposition}

\newtheorem{corollary}{Corollary}

\newtheorem{remark}{Remark}

\DeclareMathOperator*{\argmin}{arg\,min}

\newcommand{\bbP}{\mathbb{P}}

\newcommand{\possiblemerges}[1]{\mathcal{M}^{(#1)}}
\newcommand{\merge}[1]{M^{(#1)}}
\newcommand{\mergeobs}[1]{M_o^{(#1)}}
\newcommand{\mergecol}[1]{\overline{M}^{(#1)}}
\newcommand{\mergecolobs}[1]{\overline{M}_o^{(#1)}}

\newcommand{\tradmerge}[1]{M^{*(#1)}}
\newcommand{\tradmergecol}[1]{\overline{M}^{*(#1)}}

\newcommand{\existingclusters}{\mathcal{C}}
\newcommand{\existingcluT}[1]{\mathcal{C}^{(#1)}}
\newcommand{\winclu}[2]{{C}^{(#1)}_{#2}} 
\newcommand{\wincluobs}[2]{C^{(#1)}_{#2,o}}

\newcommand{\selprob}[1]{p^{(#1)}(M;\mathcal{M}_o^{(#1)},X_o)}
\newcommand{\selprobtXu}[3]{p^{(#1)}\left(M_o^{(#2)};\mathcal M_o^{(#2)},X(#3; \mathcal{A}_o^{(t)})\right)}
\newcommand{\selprobtXuchi}[3]{p^{(#1)}\left(M_o^{(#2)};\mathcal{M}_o^{(#1)},X(#3; \mathcal{A}_{\Sigma,o}^{(t)})\right)}

\newcommand{\pvalue}{\normalfont{\textbf{P}}^{(t)}}

\newcommand{\indep}{\perp \!\!\! \perp}

\newcommand{\anon}{1}

\begin{document}

\def\spacingset#1{\renewcommand{\baselinestretch}%
{#1}\small\normalsize} \spacingset{1}


\if1\anon
{
  \title{\textbf{Hierarchical Clustering With Confidence}}
  \author{Di Wu\\
Department of Statistics, University of Michigan\\
and \\
Jacob Bien\\
Department of Data Sciences and Operations,
University of Southern California\\
and\\
Snigdha Panigrahi\thanks{The author gratefully acknowledges support from NSF CAREER Award DMS-2337882.}\hspace{.2cm}\\
Department of Statistics, University of Michigan}
  \maketitle
} \fi

\if0\anon
{
  \bigskip
  \bigskip
  \bigskip
  \begin{center}
    {\LARGE\bf Hierarchical Clustering With Confidence}
\end{center}
  \medskip
} \fi

\bigskip
\begin{abstract}
Agglomerative hierarchical clustering is one of the most widely used approaches for exploring how observations in a dataset relate to each other.  
However, its greedy nature makes it highly sensitive to small perturbations in the data, often producing different clustering results and making it difficult to separate genuine structure from spurious patterns.
In this paper, we show how randomizing hierarchical clustering can be useful not just for measuring stability but also for designing valid hypothesis testing procedures based on the clustering results. 

We propose a simple randomization scheme together with a method for constructing a valid p-value at each node of the hierarchical clustering dendrogram that quantifies evidence against performing the merge.
Our test controls the Type I error rate, works with any hierarchical linkage without case-specific derivations, and simulations show it is substantially more powerful than existing selective inference approaches.
To demonstrate the practical utility of our p-values, we develop an adaptive $\alpha$-spending procedure that estimates the number of clusters, with a probabilistic guarantee on overestimation.  
Experiments on simulated and real data show that this estimate yields powerful clustering and can be used, for example, to assess clustering stability across multiple runs of the randomized algorithm.
\end{abstract}

\noindent%
\vfill

\newpage
\spacingset{1.8} 

\section{Introduction}

Agglomerative hierarchical clustering is one of the most widely used approaches for exploring relationships among observations in a dataset.
By producing nested clusters, it allows researchers to visualize how observations group together across different levels of similarity, without specifying the number of clusters in advance. 
Owing to its ease of implementation and the convenient visualization of its results, this class of clustering methods has found extensive applications in diverse fields such as genomics (\citep{Eisen1998genomic}), neuroscience (\citep{MorenoDominguez2014_HierarchicalParcellation}), ecology (\citep{Zolfaghari2019_DesertificationClusters}), network analysis (\citet{SHEN20091706,chen2009}), and text mining (\citep{Steinbach2000}).

Agglomerative hierarchical clustering works by iteratively merging similar observations to form nested clusters, producing a tree-like structure called a dendrogram.
Each leaf of the dendrogram represents an individual observation, treated as a singleton cluster, while each interior node represents a subset or cluster of observations.
At each iteration, the clustering algorithm identifies the pair of clusters with the smallest dissimilarity, defined according to a chosen linkage criterion, and merges them.
This iterative, bottom-up process continues until all observations are merged into a single cluster.
See, for example, panel (a) of Figure \ref{fig:dendrogram}, which illustrates a two-dimensional dataset of $30$ observations belonging to two true clusters, while the panel (b) shows the dendrogram obtained using complete-linkage hierarchical clustering.
In this work, we focus on agglomerative hierarchical clustering, hereafter referred to simply as hierarchical clustering.

\begin{figure}[t]
        \centering
        \includegraphics[width = \linewidth]{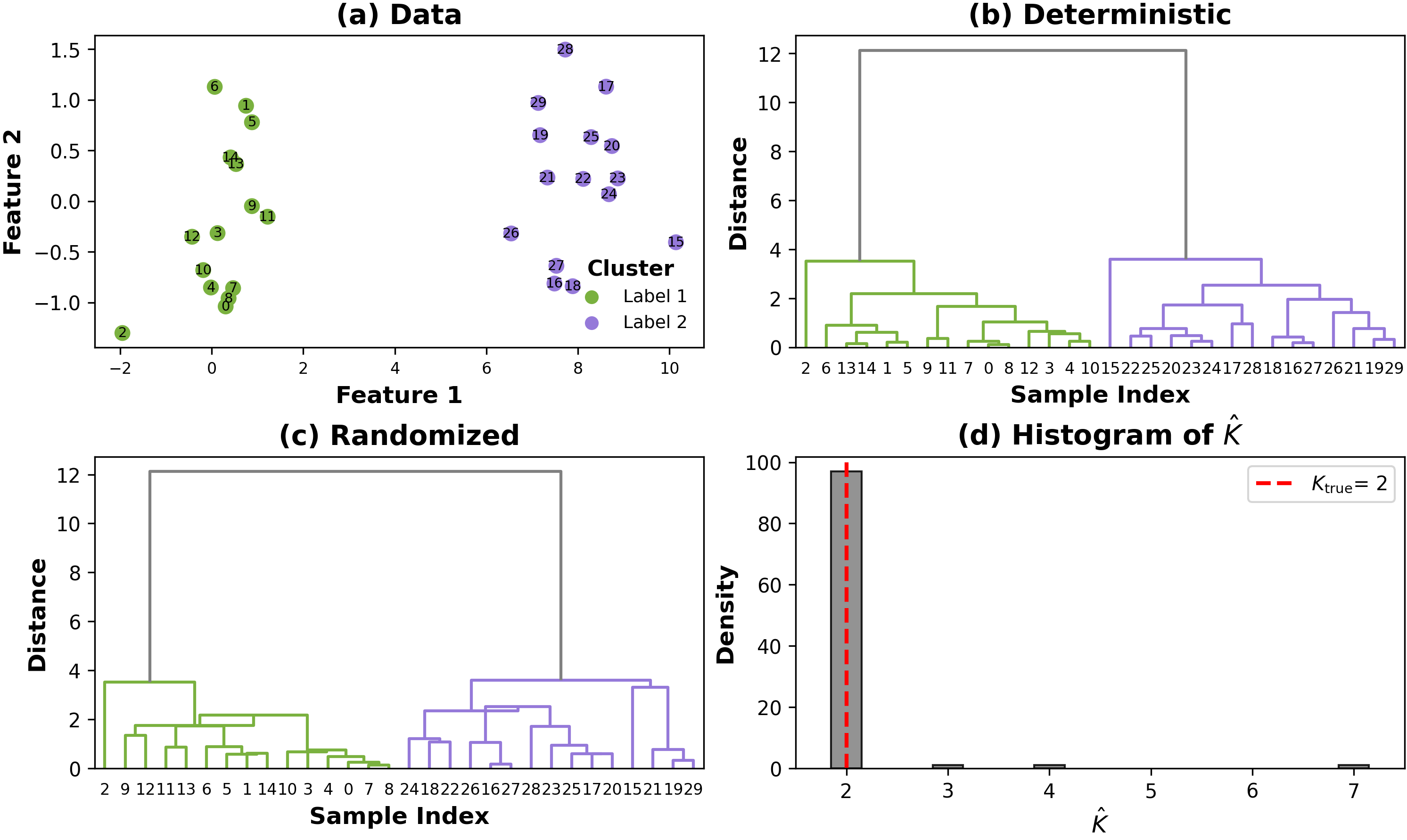}
    \caption{ (a): A two-dimensional example with $n=30$ points across two true clusters; (b): a dendrogram resulting from complete-linkage hierarchical clustering of the example data; (c): a dendrogram resulting from complete-linkage randomized hierarchical clustering of the example data (using a randomization parameter $\tau^* = 0.10$, defined later in Section \ref{sec:algorithm}); (d): A histogram for the estimated number of clusters $\hat{K}$ based on the proposed $\alpha$-spending procedure.}
    \label{fig:dendrogram}
\end{figure}

Although a complete nested clustering structure can be obtained from the tree-like dendrogram, hierarchical clustering is, at its core, a greedy algorithm.
Without principled approaches to quantify evidence against greedy merges or to identify the appropriate cut level, it becomes difficult to distinguish genuine clusters from patterns arising purely out of sampling variation.
For instance, hierarchical clustering results can be highly sensitive to small perturbations in the data due to its greedy nature.
To evaluate clustering results in the absence of labels or costly validation data, several papers (\citet{vonLuxburg2010, liu2022stability, Lange2004}) have advocated the concept of clustering stability.
The underlying idea is that genuine clustering structure should remain reproducible under small perturbations. 
To this end, \cite{Monti2003} and \cite{Hennig2007} have suggested introducing randomization into the algorithm to reveal the stability of pairwise cluster assignments.

In this paper, we show how randomizing a hierarchical clustering algorithm can be useful not just for measuring stability but also for constructing valid hypothesis testing procedures based on the clustering results.
We propose a simple randomization scheme together with a method for forming a valid p-value at each node of the dendrogram, which quantifies the statistical evidence against performing the merge.
Panel (c) of Figure \ref{fig:dendrogram} displays a dendrogram produced by one run of our randomized complete-linkage hierarchical clustering algorithm.
It matches the dendrogram in panel (b) in terms of the two real clusters in the data, while scrambling the order in which points within each cluster are merged.

Our randomization approach is motivated by recent inferential advances for decision trees \citep{Bakshi2024, bakshi2025classification}, in which conventional splits are replaced with randomized, softmax-based splits to account for the adaptive tree structure during inference.
The scope of our approach, however, is not limited to performing post-hoc inference on clusters obtained from a fixed level of the algorithm.
While hierarchical clustering does not require a prespecified number of clusters, interpreting its output in practice requires choosing an appropriate dendrogram level or cut height to extract and assess the resulting cluster structure.
To address this practical issue, we develop an adaptive $\alpha$-spending procedure based on our valid p-values to estimate the number of clusters, providing a probabilistic guarantee against overestimation.
Continuing with the example in Figure \ref{fig:dendrogram}, we apply our proposed method to estimate the number of clusters, $\widehat{K}$, over $100$ runs of our randomized algorithm; panel (d) displays the empirical histogram of $\widehat{K}$, which has a strong mode at $2$, the true number of clusters.

We next provide a detailed account of our contributions, placing them in the context of the existing literature on inference for clustering.

\subsection{Related work and contributions}

Determining whether two clusters identified by a clustering algorithm are genuinely distinct, when the same data are used both to discover and to validate these clusters, is an instance of \textit{double dipping}.
Selective inferential methods address this by explicitly accounting for the dependence introduced by reusing the same data twice.
Recent work by \cite{gao2024selective} and \cite{yun2023selective} develop valid post-hoc inference on results obtained from agglomerative clustering.
Beyond hierarchical clustering, \cite{chen2023selective} develop a selective inferential approach to validate clusters identified by the K-means algorithm.
By conditioning on the output of the clustering algorithm, these approaches yield valid p-values at a prespecified dendrogram level, controlling the Type I error rate.
\cite{Zhang2019} address a similar problem using sample splitting, which conditions on more information than necessary when compared with the aforementioned approaches.
Other work, such as \cite{Liu2008SigClust, Maitra2012Bootstrap, Suzuki2006Pvclust}, test the global null hypothesis that no clustering structure exists in the data, which is fundamentally different in spirit from the hypotheses considered in the selective inference approaches.

To test the inherently data-dependent hypotheses at each merge in the dendrogram, we condition on the clustering outcome, following the conditional framework in existing selective inference methodologies for clustering.
However, by introducing external randomization, our testing method differs considerably from previous approaches and offers several advantages, as outlined below.
\begin{enumerate}[leftmargin=*]
\item Our approach produces substantially more powerful p-values than the non-randomized methods when testing mean differences between a pair of clusters at a fixed dendrogram height.
This gain in power comes at little cost to the clustering quality when only a small amount of randomization is introduced, as seen in Figure~\ref{fig:dendrogram}, for example.
Gains in inferential power due to external randomization have been well-documented in supervised learning contexts.
Take, for example, data-carving methods in \cite{fithian2014optimal, panigrahi2023carving}, which select models using a random subsample like sample-splitting, but use the full data for inference.
Approaches using other types of randomization schemes are described in
\cite{tian2018selective, panigrahi2023approximate, Panigrahi2024Randomization, kivaranovic2024tight, huang2025selective}, among others. 

\item Furthermore, unlike randomized methods based on data-splitting, like the UV decomposition, fission, or thinning \citep{rasines2021splitting,leiner2025data, dharamshi2025generalized}, that base inference exclusively on holdout data, our approach reuses the data employed in clustering.
In particular, with a small amount of randomization, our approach produces only a slight perturbation of the clustering results on the original data while still using the full dataset for testing.

\item The p-values obtained from our randomized approach can be used in conjunction with any hierarchical clustering linkage to produce valid inference, without the need for: (i) case-by-case derivations of the conditioning set or (ii) potentially expensive Monte Carlo or resampling-based approximations to compute the p-values.
For example, in complete-linkage hierarchical clustering (which was used in Figure \ref{fig:dendrogram}), the tests in \cite{gao2024selective, yun2023selective} rely on Monte Carlo approximations to compute p-values, whereas our method does not resort to such sampling-based approximations.
\item Adding to the flexibility of our approach, the proposed randomization scheme can be readily applied to obtain adjusted distributions of test statistics for different types of tests, while accounting for the data-dependent nature of clustering outputs. Examples include tests based on the adjusted chi-squared test statistic in \cite{gao2024selective} and the adjusted F test statistic in \cite{yun2023selective}.
\end{enumerate}

We summarize our main contributions below and present an outline of the paper.
\begin{itemize}[leftmargin=*]
\item In Section \ref{sec:algorithm}, we introduce the randomized hierarchical clustering algorithm (Algorithm \ref{alg:RAC}), which performs clustering using a dissimilarity matrix that is updated at each merge.
\item In Section \ref{sec:pvals}, we construct p-values conditional on all previous merges to determine whether there is statistical evidence that a given merge should have been made.
Theorem \ref{thm:formula_density} derives the corresponding conditional distribution, and Theorem \ref{thm:cond_p_value} shows that p-values computed from this distribution control an appropriately defined Type I error.
\item In Section \ref{sec:estimateK}, we develop an adaptive $\alpha$-spending procedure for estimating the appropriate dendrogram cut height, as summarized in Algorithm \ref{alg:alpha_spending_stop}.
The Type I error control provided by our p-values at each merge allows us to allocate $\alpha$-values prudently based on the sizes of the clusters being merged. 
Theorem \ref{thm:overest:control} establishes a probabilistic guarantee, ensured by our algorithm, against overestimating the total number of clusters.
\item  Simulation results in Section~\ref{sec:simulations} demonstrate that: (i) our randomized approach yields valid p-values at each merge in the dendrogram, (ii) tests based on these p-values achieve higher power than existing selective inference methods at a fixed dendrogram height, without sacrificing clustering performance when only a small amount of randomization is used, and (iii) the proposed $\alpha$-spending procedure to estimate the dendrogram cut height yields a powerful clustering method that outperforms the standard clustering paired with the popular gap statistic \citep{tibshirani2001estimating}. 
In a real data analysis, we illustrate how the number of clusters estimated using our method can be used to assess the stability of clustering results across multiple runs of randomized Algorithm \ref{alg:RAC}.
\item Section \ref{sec:conclusion} concludes the paper with a brief summary and ideas for future research.
\end{itemize}

\section{Randomized hierarchical clustering}\label{sec:algorithm}

In traditional agglomerative hierarchical clustering, we begin with each observation in its own cluster and at each step we find the pair of clusters that are closest and merge them.  To evaluate each candidate merge $M=(C_1,C_2)$, we use a dissimilarity between the clusters $C_1$ and $C_2$ which in this context is called a {\em linkage function}, $d(M;X)$, where $X\in\mathbb R^{n\times p}$ is the data matrix.  In particular, at step $t$ one chooses a merge
\begin{equation}
\tradmerge{t}\in\arg\min_{M\in \mathcal M^{*(t)}}d(M;X),\label{eq:trad_hclust}    
\end{equation}
where $\mathcal M^{*(t)}$ is the allowable set of merges at step $t$, consisting of all pairs of clusters existing at that step.  In practice, there are two choices built into the function $d(M; X)$: (i) how to form a $n\times n$ dissimilarity matrix $D^{(1)}(X)$ between the rows of $X$ (e.g., Euclidean distance) and (ii) how to go from those base dissimilarities to the dissimilarity between a pair of clusters (e.g., taking the minimum/maximum/average dissimilarity between pairs of points in different clusters).  Depending on these choices, ties may arise in \eqref{eq:trad_hclust}, in which case there can be multiple merge sequences $\tradmergecol{t}:=(\tradmerge{1},\ldots,\tradmerge{t})$ that are equally valid.

\subsection{Proposed randomization scheme}

Our proposed algorithm relaxes the greedy nature of traditional hierarchical clustering by randomizing the merge selection.  Instead of selecting a {\em minimizing} merge $\tradmerge{t}$ at each step, we sample a {\em random} merge $M^{(t)}$ from $\mathcal M^{(t)}$, the current set of allowable merges. Notationally, we remove stars to distinguish our random sequence from the greedy algorithm's sequence.  
Furthermore, let $\mergecol{t-1}:=(M^{(1)}, \ldots, M^{(t-1)})$ denote the random sequence of previous merges until step $t-1$. Then, note that $\mathcal M^{(t)}$, the set of allowable merges at step $t$, is a deterministic function of the past merges, i.e., $\mathcal M^{(t)}:=\mathcal M(\mergecol{t-1})$.

We are now ready to describe how we sample the random merge at step $t$. 
Conditional on $\mergecol{t-1}$ and $X$ (i.e, the sequence of previous merges and the data), $M^{(t)}$ is sampled according to the probability mass function 
\begin{equation}
\begin{aligned}
\small\bbP\left(M^{(t)}=M|\mergecol{t-1}=\mergecolobs{t-1},X=X_o \right)= 
\selprob{t}:=\dfrac{\exp\left(- d(M; X_o)/\tau(\mathcal M_o^{(t)},X_o)\right)}{\displaystyle\sum_{M'\in \mathcal M_o^{(t)}}\exp\left(-d(M';X_o)/\tau(\mathcal M_o^{(t)},X_o)\right)},
\end{aligned}
\label{eq:sel_prob}
\end{equation}
where $X_o$ is the observed data matrix, $\tau$ is a user-specified function controlling the amount of external randomization in the merge selection, and $\mathcal M^{(t)}_o=\mathcal M(\mergecolobs{t-1})$ is the observed set of currently allowable merges.

   
In brief, we have replaced the $\arg\min$ over the vector $D^{(t)}(X):=\{d(M;X):M\in\mathcal M^{(t)}\}$ by sampling from the softmax of $-D^{(t)}(X)/\tau(\mathcal M^{(t)},X)$.
The randomized merging scheme introduced above closely resembles the exponential mechanism by \cite{mcsherry2007mechanism}, which is used in differential privacy to randomly sample an item from a discrete set while providing privacy guarantees.
Recent work by \cite{zhang2024winners} uses a similar type of sampling probabilities as weights in a weighted average-type procedure to obtain stability, which in turn is used to establish the asymptotic normality of cross-validated risks. 
A key distinction of our randomized approach from the broader literature on differential privacy and algorithmic stability is that we neither require nor seek to achieve any form of privacy or stability in order to construct valid testing procedures, as described later in Section \ref{sec:pvals}.

\subsection{Amount of randomization}
\label{subsec:amount:rand}

At each step of the randomized clustering algorithm, one must determine the amount of randomization used to sample the merge.
We define the user-specified function controlling the amount of external randomization as
\begin{equation}
\tau(\possiblemerges{t},X):= \tau^*\cdot \frac{1}{|\possiblemerges{t}|}\sum_{M'\in\possiblemerges{t}}d(M';X),
\label{eqn:scaletau}
\end{equation}
i.e., at step $t$, we set the amount of randomization to be a constant fraction $\tau^*$ of the average dissimilarity among the potential merges at that step.
This data-driven specification of $\tau$ ensures that the clustering structure produced by our randomized algorithm remains invariant under rescaling of the data, a property also exhibited by the deterministic version.


The smaller the value of $\tau$, the lower the amount of randomization in our clustering algorithm. 
Our next result, Proposition \ref{prop:limit}, shows that as $\tau^* \to 0$, the random sequence of merges produced by our randomized clustering algorithm, with probability approaching $1$, coincides with the sequence generated by the traditional algorithm under the same linkage criterion. 
Corroborating this result, our numerical experiments in Section \ref{subsec:clustering_quality} show that small values of $\tau^*$
lead to clustering results that closely resemble the clustering output on the full data in terms of standard clustering evaluation metrics. 

\begin{proposition}
Let $\Omega_o^{*(t)}$ be the set of possible merge sequences from running traditional hierarchical clustering \eqref{eq:trad_hclust} for $t$ steps on a fixed data matrix $X_o$, i.e.
\begin{equation}
\Omega_o^{*(t)}=\left\{(\tradmergecol{t-1}_o,\tradmerge{t}):\tradmergecol{t-1}_o\in \Omega_o^{*(t-1)},~\tradmerge{t}\in\argmin_{M\in \mathcal M(\tradmergecol{t-1}_o)}d(M;X_o)\right\}.
\label{eq:set_mergeseq}
\end{equation}
     Then, the randomized algorithm's random merge sequence $\mergecol{t}$ when presented with the same data satisfies the following:
     $$
     \lim_{\tau^* \rightarrow 0} \bbP(\mergecol{t}\in \Omega^{*(t)}_o|X=X_o) =1.
     $$
\label{prop:limit}
\end{proposition}
\vspace{-1em}
A detailed proof for this result is deferred to Appendix \ref{proof:prop1}.

Algorithm~\ref{alg:RAC} presents the proposed randomized hierarchical clustering method,  producing $K$ clusters by performing $n - K$ merges. In the special case where $K = 1$, the algorithm proceeds until all observations are merged into a single cluster.
\begin{algorithm}[htbp]
\setstretch{1.5}
\caption{Randomized Agglomerative Clustering Algorithm}
\label{alg:RAC}
\begin{algorithmic}[1]
\vspace{0.2cm}
\State \textbf{Input:} Pairwise dissimilarity matrix $D^{(1)}(X)\in\mathbb R^{n\times n}$, number of clusters $K$, linkage function $d(\cdot; X)$, parameter controlling the amount of randomization $\tau^*$
\State \textbf{Initialization:} Each data point as its own singleton cluster: 
$\existingclusters^{(1)}= \left\{\{1\}, \{2\}, \dots, \{n\}\right\}$

\For{$t\in [n-K]$}

\State Determine the set of allowable merges $\possiblemerges{t}$ from the clustering $\existingclusters^{(t)}$. For each pair of cluster $M\in \possiblemerges{t}$, compute $d(M;X)$ using linkage $d(\cdot; X)$.

\State Set $\tau(\possiblemerges{t}, X):= \tau^*\cdot \dfrac{1}{|\possiblemerges{t}|}\sum_{M\in \possiblemerges{t}}d(M;X)$



\State Sample the winning merge $\merge{t}=\left(\winclu{t}{1},\winclu{t}{2}\right)$ from probability mass function \eqref{eq:sel_prob}.
\State Merge the winning pair of clusters and update the collection of clusters: 
$$ \existingcluT{t+1} \gets \existingcluT{t}\setminus \{
    \winclu{t}{1},\winclu{t}{2}\} \cup \left\{\winclu{t}{1}\cup \winclu{t}{2}\right\}$$
\EndFor
\State \textbf{Output:} Cluster labels for each data point
\end{algorithmic}
\end{algorithm}

\section{Conditional inference: should we merge?}\label{sec:pvals}

\subsection{Null hypothesis and conditional Type I error}

In this section, we develop a hypothesis testing procedure to obtain statistical evidence on whether a merge should be performed. 
Here, we formulate the null hypothesis and discuss the \textit{conditional} Type I error rate, which guarantees that the p-values we construct for each merge remain valid conditional on the clustering output up to that merge.

We consider the following model for $n$ observations with $p$ features:
\begin{equation}\label{eq:model}
    X \sim \mathcal{MN}_{n\times p}(\mu,I_n,\sigma^2 I_p),
\end{equation}
where the rows of $\mu \in \mathbb{R}^{n\times p}$ represent the unknown mean parameters, and $\sigma^2 > 0$ denotes the noise variance, i.e., each observation $X_i \in \mathbb{R}^p$, for $i\in \{1,2,\ldots,n\}$, is an independent realization from $\mathcal{N}_p(\mu_i, \sigma^2 I_p)$. 

A particular realization of Algorithm~\ref{alg:RAC} on $X_o$ (an observed realization of the random matrix $X$), yields a fixed sequence of merges up to time $t$, which we will denote $\mergecolobs{t} = \left(\mergeobs{1}, \mergeobs{2}, \ldots, \mergeobs{t}\right).$ At step $t$, the chosen merge $\mergeobs{t} = \left(\wincluobs{t}{1}, \wincluobs{t}{2}\right)$ gives rise to the null hypothesis 
\begin{equation}\label{eq:hypothesis}
    H_0^{(t)}: \mu_{i}=\mu_{i'}  \quad \text{for all}\quad  i,i' \in \wincluobs{t}{1} \cup \wincluobs{t}{2},
\end{equation}
which states that the data points belonging to the merged pair of clusters $\wincluobs{t}{1}$ and $\wincluobs{t}{2}$ share the same mean parameter.

In what follows, we derive an exact, closed-form, and valid p-value for testing this null hypothesis, treating the variance parameter $\sigma^2$ in \eqref{eq:model} as unknown.
Before presenting our testing procedure, we first present the conditional Type I error guarantee associated with the p-values we construct for each merge.

At each merge, the p-value we construct, denoted by $\pvalue$, satisfies the following conditional Type I error guarantee under the null hypothesis in \eqref{eq:hypothesis}:
\begin{equation}
\bbP_{H_0^{(t)}}\left(\pvalue \leq \alpha \; | \;\mergecol{t} =\mergecolobs{t} \right) = \alpha, \quad \text{ for any } 0\leq\alpha\leq 1.
\label{cond:guarantee}
\end{equation}
In words, $\pvalue$ follows a uniform distribution when we restrict attention to those runs of our randomized clustering algorithm on $X$ that produce the observed sequence of merges $\mergecolobs{t}$.
As we show later, in constructing our $\alpha$-spending testing procedure, we rely on the conditional Type I error rate in \eqref{cond:guarantee}, allowing for an adaptive choice of test level $\alpha=\alpha^{(t)}(\mergecol{t-1})$ at step $t$ that depends on all previous merges up to that step.

\subsection{Test statistic and preliminaries}

\paragraph{Notations.} 
We begin by introducing the necessary notation. 
Consider the pair of merged clusters $\left(\wincluobs{t}{1}, \wincluobs{t}{2}\right)$ at step $t$, whose respective means are denoted by $\bar{X}_{\wincluobs{t}{1}}$ and $\bar{X}_{\wincluobs{t}{2}}$. 
The total number of observations in these two clusters is denoted by $N_o^{(t)} = |\wincluobs{t}{1}| + |\wincluobs{t}{2}|$.  

Define $\nu_{\mergeobs{t}}$ as an $n$-dimensional vector whose $i$-th entry is given by
\begin{equation}\label{eq:nuM}
   \left[\nu_{\mergeobs{t}}\right]_i = \mathbbm{1}\{i\in\wincluobs{t}{1}\}/|\wincluobs{t}{1}| - \mathbbm{1}\{i\in\wincluobs{t}{2}\}/|\wincluobs{t}{2}|.  
\end{equation}

We also define the matrices
\begin{align*}
    B_o^{(t)} &= \frac{\nu_{\mergeobs{t}}\nu_{\mergeobs{t}}^T}{\|\nu_{\mergeobs{t}}\|_2^2},\ \quad
    W_o^{(t)} =  \left( I_{\wincluobs{t}{1}} - \frac{1_{\wincluobs{t}{1}}1_{\wincluobs{t}{1}}^\top}{|\wincluobs{t}{1}|} \right) + \left( I_{\wincluobs{t}{2}} - \frac{1_{\wincluobs{t}{2}} 1_{\wincluobs{t}{2}}^\top}{|\wincluobs{t}{2}|} \right),
\end{align*}
where $I_C$ denotes the $n\times n$ diagonal matrix with entry $(i,i)$ set to $\mathbbm{1}\{i \in C\}$, and $1_{C}$ denotes the $n$-dimensional vector with entry $i$ set to $\mathbbm{1}\{i \in C\}$.

Given these definitions, it follows that: 
(i) $\nu_{\mergeobs{t}}^T X = \bar{X}_{\wincluobs{t}{1}} - \bar{X}_{\wincluobs{t}{2}}$ captures the differences between the means of clusters merged at step $t$; (ii) $B^{(t)}_oX$ is the projection of $X$ onto a rank-one subspace capturing the differences between the cluster means of the merged pair; (iii) $W^{(t)}_oX$ is the projection of $X$ onto a subspace with rank $N_o^{(t)}-2$ capturing the differences among observations within the two merged clusters relative to their cluster means; and (iv) the subspace that $W_o^{(t)}$ projects onto is orthogonal to that of $B_o^{(t)}$.

\paragraph{Test statistic.} 
We consider the test statistic 
\begin{equation}
\label{test:statistic}
R^{(t)}\left(X; \wincluobs{t}{1}, \wincluobs{t}{2}\right):= (N_o^{(t)}-2) \cdot \dfrac{\text{BCSS}^{(t)}}{\text{WCSS}^{(t)}},
\end{equation}
 defined as the ratio between the between-cluster sum of squares $\text{BCSS}^{(t)}$ and the within-cluster sum of squares $\text{WCSS}^{(t)}$ for the merged clusters, where 
\begin{align*}
\begin{gathered}
\text{BCSS}^{(t)}:=  \dfrac{|\wincluobs{t}{1}| |\wincluobs{t}{2}|}{N_o^{(t)}} \times \Big\|\bar{X}_{\wincluobs{t}{1}}- \bar{X}_{\wincluobs{t}{2}}\Big\|_2^2= \|B^{(t)}_oX\|^2_F, \\ 
\text{WCSS}^{(t)}:= \sum_{i\in \wincluobs{t}{1}} \Big\|X_i - \bar{X}_{\wincluobs{t}{1}}\Big\|_2^2 + \sum_{i\in \wincluobs{t}{2}} \Big\|X_i - \bar{X}_{\wincluobs{t}{2}}\Big\|_2^2 =  \|W^{(t)}_oX\|^2_F.
\end{gathered}
\end{align*}
When testing for differences in means between two predefined clusters or groups, the statistic in \eqref{test:statistic} yields the standard F test. 
However, for clusters dependent on data via the greedy merging algorithm, the null distribution of this statistic is no longer an F distribution.
Hereafter, we denote the test statistic $R^{(t)}\left(X; \wincluobs{t}{1}, \wincluobs{t}{2}\right)$ as simply $R^{(t)}$, and its observed counterpart, $R^{(t)}\left(X_o; \wincluobs{t}{1}, \wincluobs{t}{2}\right)$, as $R_o^{(t)}$.

For traditional (non-randomized) hierarchical clustering, \citet{yun2023selective} derive the distribution of this statistic under the null hypothesis \eqref{eq:hypothesis} by conditioning on the clustering output, showing that it follows a truncated F distribution.
In general, however, this null distribution is not available in closed form, as the truncation region typically lacks an explicit characterization, except in a few special cases.
Consequently, \cite{yun2023selective} employ Monte Carlo approximations to compute their p-values, which cannot be calculated exactly.
As emphasized earlier in the paper, we take a different approach to this problem by utilizing our randomization scheme to altogether bypass the need to describe the underlying conditioning event.

\paragraph{Preliminaries.} 
The distribution of $R^{(t)}$ conditional on all merges up to step $t$ is difficult to characterize, since the merges made by a clustering algorithm depend on the entire data matrix rather than solely on the test statistic of interest.
To this end, we state a few preliminaries, adopted from \cite{yun2023selective}.

Lemma \ref{lem:recX} identifies auxiliary statistics that, together with the test statistic, enable the reconstruction of the data matrix $X$. 
Consequently, the sequence of merges in a non-randomized clustering merging process is entirely determined by these statistics, whereas in our randomized algorithm, the sequence is determined by both these statistics and the additional randomization introduced into the merging process.

Define the auxiliary statistics $\mathcal A^{(t)}:=(\eta^{(t)},\gamma^{(t)},\Delta^{(t)}, \Gamma^{(t)})$, where
\begin{align}
\begin{gathered}
\eta^{(t)} = dir(B^{(t)}_oX) , \  \gamma^{(t)} = dir(W^{(t)}_oX)\\
\Delta^{(t)} = \|B^{(t)}_oX\|_F^2 + \|W^{(t)}_oX\|^2_F,  \ \Gamma^{(t)} = (I_n - B^{(t)}_o-W^{(t)}_o)X.
\end{gathered}
\label{add:statistics}
\end{align} 
The observed auxiliary statistics, $\mathcal{A}^{(t)}_o = \left(\eta^{(t)}_o, \gamma^{(t)}_o, \Delta^{(t)}_o, \Gamma^{(t)}_o\right)$, are defined as above but with $X$ replaced by $X_o$.
\begin{lemma}
For $r\in\mathcal{R^+}$, define the function
\begin{equation*}
X(r;\mathcal{A}^{(t)}_o) = \sqrt{\Delta^{(t)}_o}\cdot\left(\eta^{(t)}_o\sqrt{\frac{r}{N_o^{(t)}-2+r}} + \gamma^{(t)}_o\sqrt{\frac{N_o^{(t)}-2}{N_o^{(t)}-2+r}}\right) + \Gamma^{(t)}_o,
\end{equation*}
where $N_o^{(t)} = |\wincluobs{t}{1}| + |\wincluobs{t}{2}|$ and $\mathcal{A}_o^{(t)} = \left(\eta^{(t)}_o, \gamma^{(t)}_o, \Delta^{(t)}_o, \Gamma^{(t)}_o\right)$ is the observed auxiliary statistics.
It holds that 
$X(R^{(t)}_o;\mathcal{A}^{(t)}_o)= X_o$,
where $R^{(t)}_o$ is as defined in \eqref{test:statistic}.
\label{lem:recX}
\end{lemma}

The proof for this result can be found in Appendix \ref{proof:lem2}.

The next lemma 
shows that the test statistic $R^{(t)}$ given the auxiliary statistics has an $F$ distribution under $H_0^{(t)}$, provided that the clusters $(\wincluobs{t}{1}$, $\wincluobs{t}{2})$ are predefined, i.e., their dependence on data is ignored.

\begin{lemma}[\emph{Based on \cite{yun2023selective}}]
\label{lem:inf_target_dist}
Consider the null hypothesis $H_0^{(t)}$ in (\ref{eq:hypothesis}), assuming that the clusters $(\wincluobs{t}{1}$, $\wincluobs{t}{2})$ are predefined.
Then, the distribution of
$$
R^{(t)}  \mid \mathcal A^{(t)}=\mathcal A_o^{(t)}
$$ 
coincides with that of an $F_{p,(N_o^{(t)}-2)p}$ random variable. 
\end{lemma}

A detailed proof for this Lemma is deferred to Appendix \ref{proof:lem:inf_target_dist}.

Our main result, presented in the next section, further conditions this distribution on the sequence of merges made by our randomized clustering algorithm to obtain the correct distribution of the test statistic and construct a p-value that satisfies \eqref{cond:guarantee}.

\subsection{Constructing a conditionally valid p-value}
To construct a valid p-value for the null hypothesis $H_0^{(t)}$, we characterize the conditional distribution of
\begin{equation}
\label{eq:target_R}
    R^{(t)} \Big\lvert \left\{\mathcal A^{(t)}=\mathcal A_o^{(t)}, \mergecol{t} = \mergecolobs{t} \right\}.
\end{equation}

Theorem \ref{thm:formula_density} derives this conditional distribution, accounting for the highly data-dependent nature of the clustering process by conditioning on $\left\{\mergecol{t} = \mergecolobs{t}\right\}$, the sequence of  merges up to step $t$.
The additional conditioning on the auxiliary statistics reduces the calculation of this conditional distribution to computing one-dimensional integrals, eliminating the need to marginalize over them in the density function.

\begin{theorem}
\label{thm:formula_density}
Under $H_0^{(t)}$ in \eqref{eq:hypothesis}, the cumulative distribution function (CDF) of the conditional distribution of the test statistic $R^{(t)}$ in \eqref{eq:target_R}, evaluated at $r$, is given by
 \[\mathbb{F}^{(t)}(r;\mathcal{A}^{(t)}_o, \mergecolobs{t}) = \frac{\int_0^r\ell_{F_{p,(N_o^{(t)}-2)p}}(r')\times \prod_{s=1}^{t} \selprobtXu{s}{s}{r'}dr'}{\int_0^\infty\ell_{F_{p,(N_o^{(t)}-2)p}}(r')\times \prod_{s=1}^{t}\selprobtXu{s}{s}{r'}dr'},\]
 where $\ell_{F_{p,(N_o^{(t)}-2)p}}(r)$ denotes the density of an $F_{p,(N_o^{(t)}-2)p}$ random variable, and 
 $$\selprobtXu{s}{s}{r'}= \mathbb{P}\left(M^{(s)}= M_o^{(s)} |\mergecol{s-1}=\mergecolobs{s-1},X= X(r'; \mathcal{A}_o^{(t)})\right)$$
  is the sampling probability that $M^{(s)}= M_o^{(s)}$ given $X= X(r'; \mathcal{A}_o^{(t)})$ and the merge history $\mergecol{s-1}=\mergecolobs{s-1}$, for $s\in \{1,2,\ldots, t\}$, as defined in \eqref{eq:sel_prob}. 
\end{theorem}

A detailed proof for this Theorem is deferred to Appendix~\ref{proof:thm1}.

Next, we present the conditional p-value, denoted by $\pvalue = \pvalue\left(R^{(t)};\mathcal{A}^{(t)}\right)$, which is a function of the test statistic and the auxiliary statistics.
Let 
$\pvalue_o = \pvalue\left(R^{(t)}_o; \mathcal{A}^{(t)}_o\right)$
denote its value computed on the observed data $X_o$, defined as
\begin{equation}\label{eq:pval}
 \pvalue_o = \mathbb{P}_{H_0^{(t)}}\left(R^{(t)}\geq R^{(t)}_o \; \Big\lvert \;\mergecol{t} = \mergecolobs{t},\mathcal A^{(t)}=\mathcal A_o^{(t)}\right),
\end{equation}
where $R^{(t)}_o$ is the observed value of the test statistic $R^{(t)}$.

\begin{corollary}
It holds that $\pvalue_o = 1-\mathbb{F}^{(t)}(R^{(t)}_o;\mathcal{A}^{(t)}_o,\mergecolobs{t})$.
\label{cor:pval}
\end{corollary}

\begin{proof}
 This claim follows directly from the definition of $\pvalue_o$ and the conditional distribution derived in Theorem \ref{thm:formula_density}.
\end{proof}

In Theorem \ref{thm:cond_p_value}, we establish that the test based on this p-value controls the conditional Type I error, defined in \eqref{cond:guarantee}, at the nominal level $\alpha$.

\begin{theorem}\label{thm:cond_p_value}
We have that
\begin{equation*}
    \bbP_{H_0^{(t)}}\left(\pvalue(R^{(t)};\mathcal{A}^{(t)})\leq \alpha \; | \;\mergecol{t} =\mergecolobs{t}\right) = \alpha, \quad \text{ for any} \quad 0\leq\alpha\leq 1.
\end{equation*}
\end{theorem}
A proof for this result is included in Appendix \ref{proof:thm2}.

\subsection{Computational benefits of our randomization scheme}

Here we highlight the benefits of our randomization scheme by contrasting the p-values in \eqref{eq:pval} with those derived from a truncated F distribution using the traditional clustering algorithm.
Introducing some notations, let $\mathbb{T}^{(t)}(r;\mathcal{A}^{(t)}_o,\mergecolobs{t})$ denote the CDF of an $F_{p,(N^{(t)}_o-2)p}$ random variable truncated to the same observed sequence of merges $\left\{\mergecol{t} = \mergecolobs{t}\right\}$, which is given by
\begin{equation}
  \mathbb{T}^{(t)}(r;\mathcal{A}^{(t)}_o,\mergecolobs{t}) = \frac{\int_0^r\ell_{F_{p,(N_o^{(t)}-2)p}}(r')\times \prod_{s=1}^{t} \mathbbm{1}\{\mergeobs{s} \in \Psi^{*(s)}(r')\}dr'}{\int_0^\infty\ell_{F_{p,(N_o^{(t)}-2)p}}(r')\times \prod_{s=1}^{t}\mathbbm{1}\{\mergeobs{s} \in \Psi^{*(s)}(r')\}dr'},  
  \label{eqn:hard:trunc}
\end{equation}
where $\Psi^{*(s)}(r') =\displaystyle\argmin_{M \in \mathcal{M}_o^{(s)}} d(M;X(r';\mathcal{A}^{(t)}_o))$.  
This yields the p-value $\pvalue_{\text{NR},o}=1 - \mathbb{T}^{(t)}(R^{(t)}_o;\mathcal{A}^{(t)}_o,\mergecolobs{t})$, constructed without external randomization, which coincides with the p-value in Theorem~2 of \cite{yun2023selective} if the sequence of merges leading to the merge $\mergeobs{t}$ at step $t$ is unique.

Proposition~\ref{prop:pvals_conv} shows that, as the amount of randomization becomes smaller, i.e., as $\tau^*\rightarrow 0$, our p-values converge to $\pvalue_{\text{NR},o}=1 - \mathbb{T}^{(t)}(R^{(t)}_o;\mathcal{A}^{(t)}_o,\mergecolobs{t})$.

\begin{proposition}
It holds that, as $\tau^* \rightarrow 0$,
$\pvalue_o(\tau^*) {\rightarrow }\pvalue_{\text{NR},o}$.
\label{prop:pvals_conv}
\end{proposition}

A proof for this result is deferred to Appendix \ref{proof:prop2}.

The p-value $\pvalue_{\text{NR},o}$ is, however, not available in closed form, since the event $\left\{\mergecol{t} = \mergecolobs{t}\right\}$ (see the indicators  in the numerator and denominator of \eqref{eqn:hard:trunc}) lacks an analytical description except for certain special linkages and cases.
For example, \cite{yun2023selective} uses Monte Carlo approximations to compute p-values when more than two clusters are present, i.e., at any intermediate step of the merging process.
Similarly, \cite{gao2024selective} rely on Monte Carlo approximations to compute p-values for complete-linkage clustering, and do not compute p-values for linkages such as the minimax linkage or for dissimilarities outside the class of squared Euclidean distances.
Our p-values, on the other hand (see the CDF in Theorem~\ref{thm:formula_density}), do not require an explicit, linkage- and dissimilarity-specific characterization of the conditioning event. 
Instead, the correction for double dipping is derived directly from the sampling probabilities of the observed sequence of merges, defined in \eqref{eq:sel_prob}, which already encode the linkage dissimilarity and are available in closed-form.

Finally, in Appendix \ref{subsec:non_spherical}, we develop conditional p-values for the type of test proposed by \cite{gao2024selective} assuming a normal model with a known non-spherical covariance matrix $\Sigma$.
By using a different decomposition of the data matrix into the test statistic and appropriately defined auxiliary statistics, the CDF of the relevant conditional distribution is derived from the sampling probabilities in our randomization scheme, analogous to the result in Theorem \ref{thm:formula_density}.
This CDF yields valid p-values with the conditional Type I error guarantee in \eqref{cond:guarantee}.
\section{Using conditional p-values to choose the number of clusters}
\label{sec:estimateK}


We develop an adaptive $\alpha$-spending procedure \citep{foster2008alpha, aharoni2014generalized} that leverages the conditionally valid p-values from the previous section to estimate the number of clusters.
Specifically, we begin with a pre-defined sequence $\{\alpha_1, \dots, \alpha_{n-1}\}$ satisfying $\sum_{j=1}^{n-1} \alpha_j = \alpha$, where $(n-1)$ is the total number of merges.
At each merge in the dendrogram, we compute the p-value $\pvalue_o$ for the null hypothesis specified in \eqref{eq:hypothesis} and compare $\pvalue_o$ to a threshold $\alpha^{(t)}$ selected from $\mathcal{S}$, a subset of unused $\alpha$ values at step $t$. The procedure terminates as soon as there is sufficient evidence against a merge made by Algorithm \ref{alg:RAC}. 

The choice of $\alpha^{(t)}$ is allowed to depend adaptively on the sizes of the clusters being merged, and, by design, also on the $\alpha$-budget already allocated up to time $t$, which permits more liberal thresholds for merges involving larger clusters, leading to a more powerful procedure.

The full procedure is summarized in Algorithm \ref{alg:alpha_spending_stop}. 
In our implementation, we introduce a minimum cluster size cutoff $n_{\text{min}}$ such that if any cluster has size below this threshold, the corresponding merges are automatically accepted without testing. 
Additionally, if the smaller cluster size lies between $n_{\text{min}}$ and the $\alpha$-spending threshold cluster size $n^*$
 at any step $t$, we assign a smaller $\alpha^{(t)}$ from the list of unused elements of $\mathcal{S}$. 
If $n^*$ is chosen to be less than or equal to $n_{\text{min}}$, this step is skipped. 
This type of adaptive allocation of $\alpha$-values is consistent with our goal of using more liberal thresholds for merges involving larger clusters with the goal of obtaining a more powerful procedure.

We next develop a probabilistic guarantee for our algorithm against overestimation. 
Suppose $K^*$ is the true number of clusters in the data, and let $t^* = n - K^* + 1$. 
Denote by $\mathcal{T}$ the event that the true cluster structure is preserved when the randomized hierarchical clustering algorithm is executed up to step $t^*$, i.e., that no cross-cluster merges occur before $t^*$.

\begin{theorem}
Suppose Algorithm \ref{alg:alpha_spending_stop} is executed with total significance level $\alpha$ and a pre-specified sequence ${\alpha_1,\alpha_2,\dots,\alpha_{n-1}}$ satisfying $\sum^{n-1}_{j=1} \alpha_j = \alpha$. 
Then the $\alpha$-spending procedure summarized in Algorithm \ref{alg:alpha_spending_stop} guarantees $\bbP\left(\{\widehat{K}> K^*\} \cap \mathcal{T} \right)\leq \alpha$.

\label{thm:overest:control}
\end{theorem}

A proof for this result is included in Appendix \ref{proof:thm3}. 
It proceeds by demonstrating that the proposed sequential testing procedure satisfies family-wise error rate (FWER) control, which is shown in Lemma \ref{lem:FWER}.

Note that Theorem \ref{thm:overest:control} guarantees that we can control the probability of overestimating $K^*$. 
This comparison between $\widehat{K}$ and $K^*$ is meaningful only when the true clustering is obtained along the hierarchical clustering path. 
Put another way, this result ensures that we are unlikely to split true clusters unnecessarily.

The numerical experiments on simulated and real data in Sections \ref{sec:chooseK_experiments} and \ref{sec:realdata} investigate the potential of the proposed $\alpha$-spending procedure for estimating the dendrogram cut height. 
The resulting estimate of the number of clusters can then be used to assess standard metrics for clustering stability across repeated runs of our randomized algorithm or clustering algorithms with other forms of randomization. 
Across a variety of simulation settings, we find that our testing procedure for estimating the number of clusters yields a powerful clustering method that outperforms standard (i.e., non-randomized) clustering approaches paired with the popular gap statistic \citep{tibshirani2001estimating}.

\begin{algorithm}[htbp]
\setstretch{1.5}
\caption{Adaptive Alpha-Spending with Early Stopping at First Rejection}
\label{alg:alpha_spending_stop}
\begin{algorithmic}[1]
\State \textbf{Input:} Significance level $\alpha$, pre-defined sequence $\{\alpha_1, \alpha_2, \dots, \alpha_{n-1}\}$ such that $\sum_{j=1}^{n-1} \alpha_j = \alpha$, merge sequence $\mergecolobs{n-1} = \{(\wincluobs{t}{1},\wincluobs{t}{2})\}_{t=1}^{n-1}$ obtained from Algorithm \ref{alg:RAC} with randomization level $\tau^*$, testable cluster size cutoff $n_{\text{min}}$, $\alpha$-spending threshold cluster size $n^*$
\State \textbf{Initialize:} Available set $\mathcal{S} =\{\alpha_1, \alpha_2, \dots, \alpha_{n-1}\}$
\For{$t = 1$ to $n-1$}
    \If{$\min(|\wincluobs{t}{1}|, |\wincluobs{t}{2}|)\leq n_{\text{min}}$}
    \State No test is performed, and we proceed to the next merge
    \Else
    \If{$\min(|\wincluobs{t}{1}|, |\wincluobs{t}{2}|) \leq n^* $}
        \State Select test level $\alpha^{(t)} \gets \min(\mathcal{S})$ \Comment{More conservative for small clusters}
    \Else
        \State Select test level $\alpha^{(t)} \gets \max(\mathcal{S})$ \Comment{More power for larger clusters}
    \EndIf
    \State Test the null hypothesis 
    in \eqref{eq:hypothesis} by constructing the p-value $\pvalue_o$ in \eqref{eq:pval} 
    \EndIf
    \State Update available alpha set: $\mathcal{S} \gets \mathcal{S} \setminus \{\alpha^{(t)}\}$
    \If{$\pvalue_o  < \alpha^{(t)}$}
        \State Set estimated number of clusters $\widehat{K} = n - t + 1$
        \State \textbf{Output:} $\widehat{K}$
        \State \textbf{Break}
    \EndIf
\EndFor
\State \textbf{If no rejection:} Set $\widehat{K} = 1$
\State \textbf{Output:} $\widehat{K}$
\end{algorithmic}
\end{algorithm}

\section{Simulation and real data analysis}
\label{sec:simulations}

We turn now to empirical evaluations of the methodology developed in the previous sections.  We begin, in Section~\ref{subsec:clustering_quality}, by studying the effect of different levels of randomization on the quality of a clustering.  In Section~\ref{sec:validitypower}, we evaluate the validity and power of our proposed conditional p-value.  Section~\ref{sec:chooseK_experiments} assesses the effectiveness of proposed procedure for choosing $K$.  Finally, Section~\ref{sec:realdata} applies our method in a real-world scenario.
Software implementations of our tests are available in Python at \url{https://github.com/judywu4800/SI_HierarchicalClustering} and in R at \url{https://github.com/jacobbien/randhclust-project}.

\subsection{How does randomization affect clustering quality?}
\label{subsec:clustering_quality}

A natural question is how the introduced randomization impacts the quality of the resulting clustering.
To investigate this, we evaluate clustering performance across different values of $\tau$ in Algorithm \ref{alg:RAC} using the following two standard metrics: (i) the ratio between the within-cluster sum of squares (WCSS) and the total sum of squares (TSS), (ii) the Adjusted Rand Index (ARI), which measures the agreement between the estimated clustering $\widehat{C} = \{\widehat{C}_1, \dots, \widehat{C}_K\}$ and the true labels $C = \{C_1, \dots, C_{K'}\}$ of $n$ observations. The exact definitions of the two metrics are provided in Appendix~\ref{subsubsec:clustering_metrics}. Smaller WCSS/TSS values indicate more well-defined and compact clusters, while the ARI ranges from $-1$ to $1$, with values close to $1$ indicating near perfect agreement between the estimated and true clustering.

We generate 500 synthetic datasets, each consisting of $n = 30$ observations in $p = 2$ dimensions, drawn from $\mathcal{N}(\mu, \sigma^2 I_p)$ with variance $\sigma^2 = 1$. 
The true number of clusters in this setting is $K = 2$. 
Following the simulation setup in \citep{yun2023selective}, the data observation are divided into two clusters 
\begin{equation}\label{eq:2clutsters}
    \mu_1 = \dots = \mu_{n/2} = \begin{bmatrix}
        0\\0
    \end{bmatrix}, \mu_{n/2+1}=\dots= \mu_{n} = \begin{bmatrix}
        \delta\\0
    \end{bmatrix}
\end{equation}
where $\delta > 0$ controls the signal strength.
Each cluster contains $15$ observations. 

First, we apply complete-linkage randomized clustering with the true number of clusters, $K=2$ and $\delta = 6$, and compare the clustering performance for $\tau^* \in \{0, 0.025, 0.05, 0.1, 0.25, 0.5, 1, 5\}$. 
Note that $\tau^* = 0$ corresponds to the traditional agglomerative clustering algorithm, depicted as RC(0), while the nonzero values of $\tau^*$, in ascending order, correspond to randomization levels $1$-$7$, denoted RC(1) through RL(7). 

\begin{figure}[h]
        \centering
        \includegraphics[width = \linewidth]{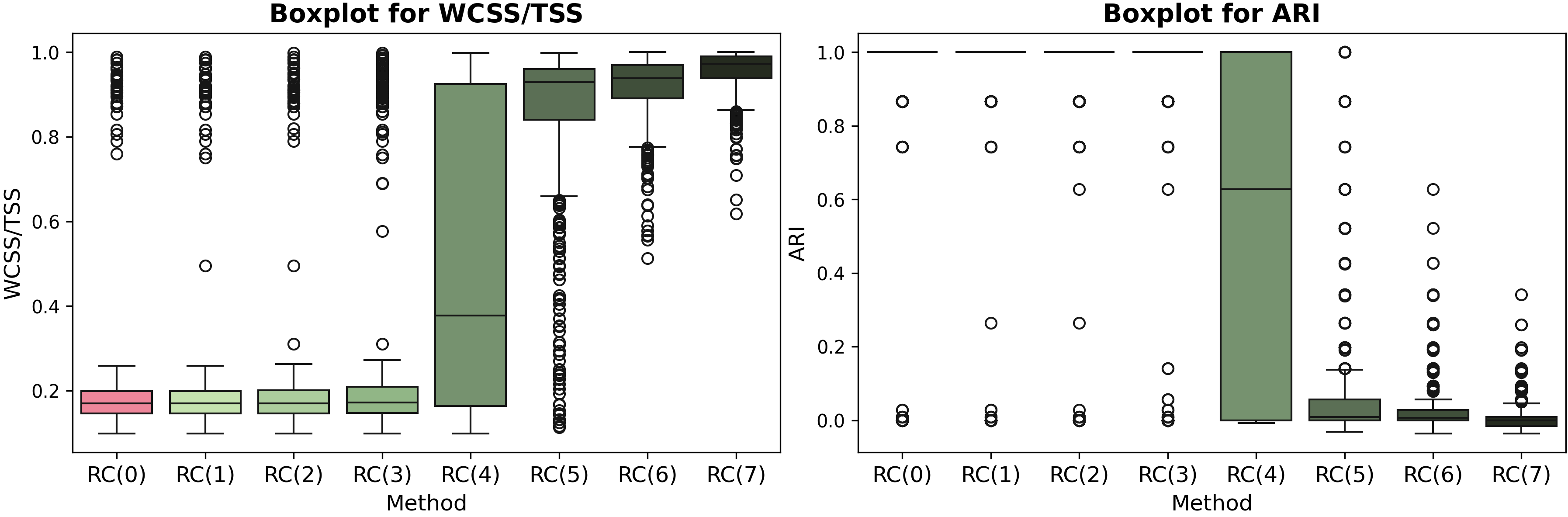}
    \caption{Comparison of clustering quality metrics under varying levels of randomization $\tau^*$. (Left): Boxplots of the ratio between within-cluster sum of squares (WCSS) and total sum of squares (TSS), showing how cluster compactness changes with $\tau^*$. (Right): Boxplots of the Adjusted Rand Index (ARI), measuring agreement with the true clustering, which declines as $\tau^*$ increases.} 
    \label{fig:clustering_quality}
\end{figure}

The results, shown in Figure \ref{fig:clustering_quality}, indicate that for small values of $\tau^*$, the randomized clustering yields results that are comparable to the deterministic agglomerative clustering in terms of both compactness of the clustering structure and label recovery. 
Consistent with expectations, as $\tau^*$ increases, the external randomness begins to dominate the merging process, resulting in less stable clusters and a clear decline in clustering quality, as reflected by higher WCSS/TSS and lower ARI.
For example, a sudden decline in clustering quality is observed for RC(4), corresponding to $\tau^* = 0.25$. 
These patterns are consistent across all three metrics used to evaluate clustering performance.

We also investigate the clustering performance under varying signal strengths $\delta$, with the corresponding results reported in Appendix~\ref{subsubsec:clustering_qual_add}.

\subsection{Evaluating p-value validity and power}
\label{sec:validitypower}

We now evaluate the validity and power of the proposed conditional p-value, focusing on its null distribution, Type I error control, and power under alternatives.

\paragraph{Type I error control.} To demonstrate Type I error control, as guaranteed by Theorem \ref{thm:cond_p_value}, we plot the ECDF of the randomized p-values $\pvalue_o$ against the uniform reference distribution and summarize the empirical Type I error rates using boxplots.

We conduct two main experiments: (i) we first show that our test  yields valid p-values for any degree of randomization, while the naive p-values suffer from the double-dipping issue and consequently yield invalid tests; (ii) 
we then compare the ECDF of the p-values obtained at randomization level $\tau^* = 0.10$ (i.e., RC(3)) with the adjusted chi-squared test from \cite{gao2024selective}, which uses the plug-in $\sigma^2$ estimate $\widehat{\sigma}^2= \frac{1}{(n-1)p}\sum_{i=1}^n\|X_i - \bar{X}\|_2^2$, and with the adjusted F test from \cite{yun2023selective}. For brevity, the detailed results of the second experiment are deferred to Appendix~\ref{subsec:valdity_additional}.

To obtain the ECDF, we perform 2000 independent trials, generating the rows of the data matrix $X$ in each trial as independent draws: $X_i \sim \mathcal{N}(\mu_i,\sigma^2 I_p), \quad i\in \{1,2,\ldots, n\}$.
For each $i$, the mean vectors are set to $\mu_i = 0_p$, so that any selected null hypothesis is true.
We fix $\sigma = 1$, $n = 30$, and $p = 10$.  
Alignment with the $y = x$ line in the ECDF plot indicates the validity of the p-values.
To evaluate Type I error control, we adopt the same data generation process as before. 
For each repetition, we perform $200$ trials to estimate the empirical Type I error rate, and we repeat this process $100$ times to construct the boxplots. 

In the first experiment, we perform complete-linkage hierarchical clustering with RC(1)-RC(7) until the data are partitioned into $K=2$ clusters, respectively, and compute p-values for the resulting clusters.
As shown in Figure \ref{fig:validity_random}, the proposed method produces p-values that are uniformly distributed under the null hypothesis and successfully controls the Type I error rate, at all levels of randomization. 
By contrast, the naive method yields excessively small p-values with a severely inflated Type I error rate.

\begin{figure}[h]
  \centering
  \includegraphics[width = 0.8\textwidth, height=0.4\textheight, keepaspectratio]{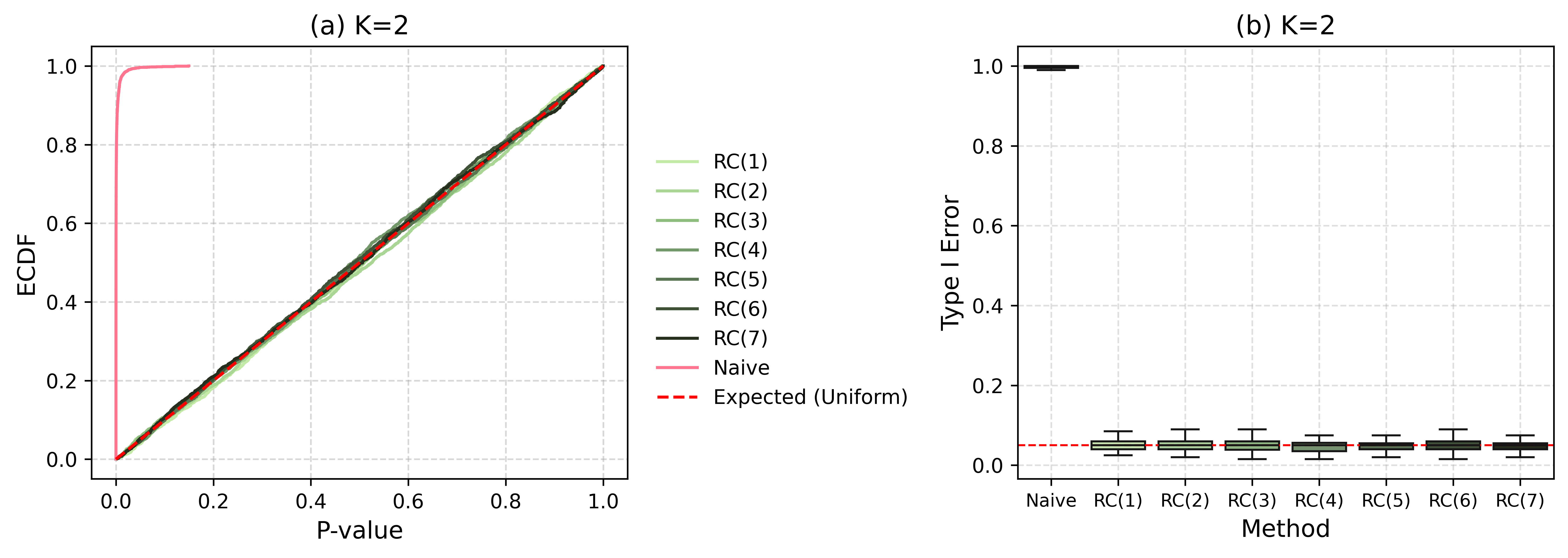}

  \caption{Comparison of the p-value ECDFs and Type I error rates simulated under the null hypothesis. (Left) ECDF plots for the proposed and baseline methods. (Right) Boxplots of the Type I error rates across different methods.}
  \label{fig:validity_random}
\end{figure}

In the second experiment, we compare the proposed p-values against the p-values from the selective inferential approaches in  \citep{gao2024selective} and \citep{yun2023selective}. 
The results confirm that the proposed method attains comparable validity.
The detailed results are deferred to Appendix \ref{subsec:valdity_additional}.

\paragraph{Power under alternatives.}
To evaluate power, we generate data from $K=2$ true clusters according to \eqref{eq:2clutsters} with $n = 30$ and $\sigma = 1$. 
We consider nine evenly spaced values of $\delta \in [1,10]$ and simulate $2000$ datasets for each $\delta$. 
For each dataset, we conduct the hypothesis test in \eqref{eq:hypothesis} at a significance level of $\alpha = 0.05$.

We compute the probability of rejecting the alternative as a function of effect size, which we define as $\text{ES}:= d(M;\mu)$,
where $d(M;\mu)$ is the dissimilarity between the true cluster means $\mu\in\mathbb R^{n\times p}$, based on the linkage function employed in the clustering algorithm.  
We note that $\text{ES}=0$ under the null for any linkage, i.e. $d(M;\mu)=0$.

To estimate the rejection probability as a function of effect size, we bin the simulated outcomes into $10$ equally spaced intervals based on their effect sizes. 
Within each bin, we calculate the proportion of rejections to obtain the estimated power.
Confidence intervals are constructed for each bin using a normal approximation to the binomial distribution.
Figure \ref{fig:power} summarizes the empirical power of the proposed randomized method with $\tau^* = 0.10$ (i.e., RC(3)) across different linkage functions and true cluster numbers $K=2$, and compares it against the two other existing selective inference methods.

\begin{figure}[h]
    \centering
    \includegraphics[width=\linewidth]{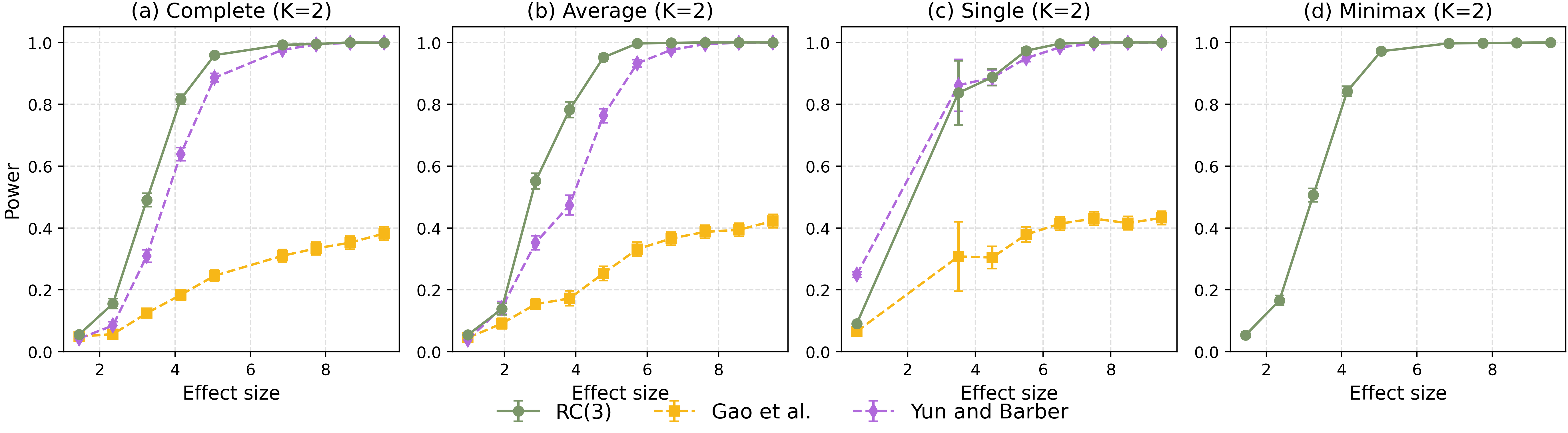}
    \caption{Empirical power curves as a function of effect size for the proposed randomized method with $\tau^* = 0.10$, compared with the two selective inference approaches under varying choice of linkage functions (complete, single, average and minimax) and true number of clusters $K=2$.}
    \label{fig:power}
\end{figure}

As evident from Figure \ref{fig:power}, the proposed method consistently attains higher power than existing selective inference approaches, with the only exception being scenarios involving smaller effect sizes under single linkage clustering.

\subsection{Evaluating our procedure for choosing the number of clusters}
\label{sec:chooseK_experiments}

We next evaluate the empirical performance of our estimation procedure for determining the number of clusters, as presented in Algorithm \ref{alg:alpha_spending_stop}, using the conditional p-values developed in Section \ref{sec:pvals}.
We first demonstrate that the sequential testing procedure achieves valid family-wise error rate (FWER) control under the global null, and then compare its ability to estimate the number of clusters with that of the widely used gap statistic proposed in \cite{tibshirani2001estimating} across a wide range of settings.

In our experiments, we employ an exponentially decaying $\alpha$-sequence of the form $\alpha_j \propto \exp{(-0.5j)}$ for $i = 1,\dots, n-1$, normalized such that the total testing budget sums to $\alpha = 0.05$. 
We set the cutoffs for the minimum testable cluster size and for deciding between a more liberal versus conservative choice of $\alpha$ in Algorithm \ref{alg:alpha_spending_stop} to $n_{\text{min}} = 0.1\times n $ and $n^* = 0.4\times n$, respectively. 
Throughout all experiments, the randomization level is fixed at $\tau^* = 0.10$, corresponding to RC(3).

To assess FWER control, we generate data under the null model in \eqref{eq:hypothesis} with $n=30$, $p=2$, and $\sigma=1$, so that no clustering structure is present. In this setting, selecting $\widehat{K}>1$ is equivalent to rejecting at least one true null hypothesis in the merge sequence. The empirical FWER is therefore the proportion of repetitions for which $\widehat{K}>1$. As shown in Table \ref{tab:fwer}, across $2000$ independent repetitions, the proposed procedure, under various randomization levels, keeps this fraction below the target level $\alpha = 0.05$. 
By contrast, Algorithm \ref{alg:alpha_spending_stop} with naive p-values (denoted Naive) fails to maintain this control.

\begin{table}[ht]
\centering
\caption{Empirical FWER for Naive and Randomized methods.}\label{tab:fwer}
\begin{tabular}{c|cccccccc}
\hline
Method  & Naive & RC(1) & RC(2) & RC(3) & RC(4) & RC(5) & RC(6) & RC(7) \\
\hline
FWER & 0.6080 & 0.0080 & 0.0070 & 0.0055 & 0.0065 & 0.0055 & 0.0050 & 0.0045 \\
\hline
\end{tabular}
\end{table}

We next evaluate how accurately the method recovers the number of clusters when true clustering structure is present. 
We compare our approach with the results from using the gap statistic in \cite{tibshirani2001estimating}, which selects $\widehat{K}$ by assessing whether the clusters obtained from the observed data are significantly tighter than those formed in reference datasets lacking structure.
Specifically, for each candidate $K$, the gap statistic computes
$\text{Gap}(K) = \mathbb{E}_{\text{ref}}[\log W_K] - \log(W_K)$,
where $W_K$ denotes the within‐cluster sum of squares. 
The selected number of clusters is then defined as
$\widehat{K}_{\text{Gap}} = \min\{K:\text{Gap}(K)\geq\text{Gap}(K+1) -s_{K+1}\},$
where $s_{K}$ is the standard deviation adjustment term defined in \citep{tibshirani2001estimating}.

We use data with $n=30$, $\sigma=1$, and $K^* = 3$, generated from three equidistant clusters in \eqref{eq:3clutsters} with between-cluster separation $\delta \in \{4,6,8,10,12,14\}$. 
For each setting, $100$ datasets are simulated, and the selected values of $\widehat{K}$ are summarized using histograms in Figure \ref{fig:paired_histogram}. 
Each histogram displays the frequency of the different $\widehat{K}$ values obtained across repeated experimental runs for both methods.

\begin{figure}[h]
    \centering
    \includegraphics[width=0.9\linewidth, height=0.7\textheight, keepaspectratio]{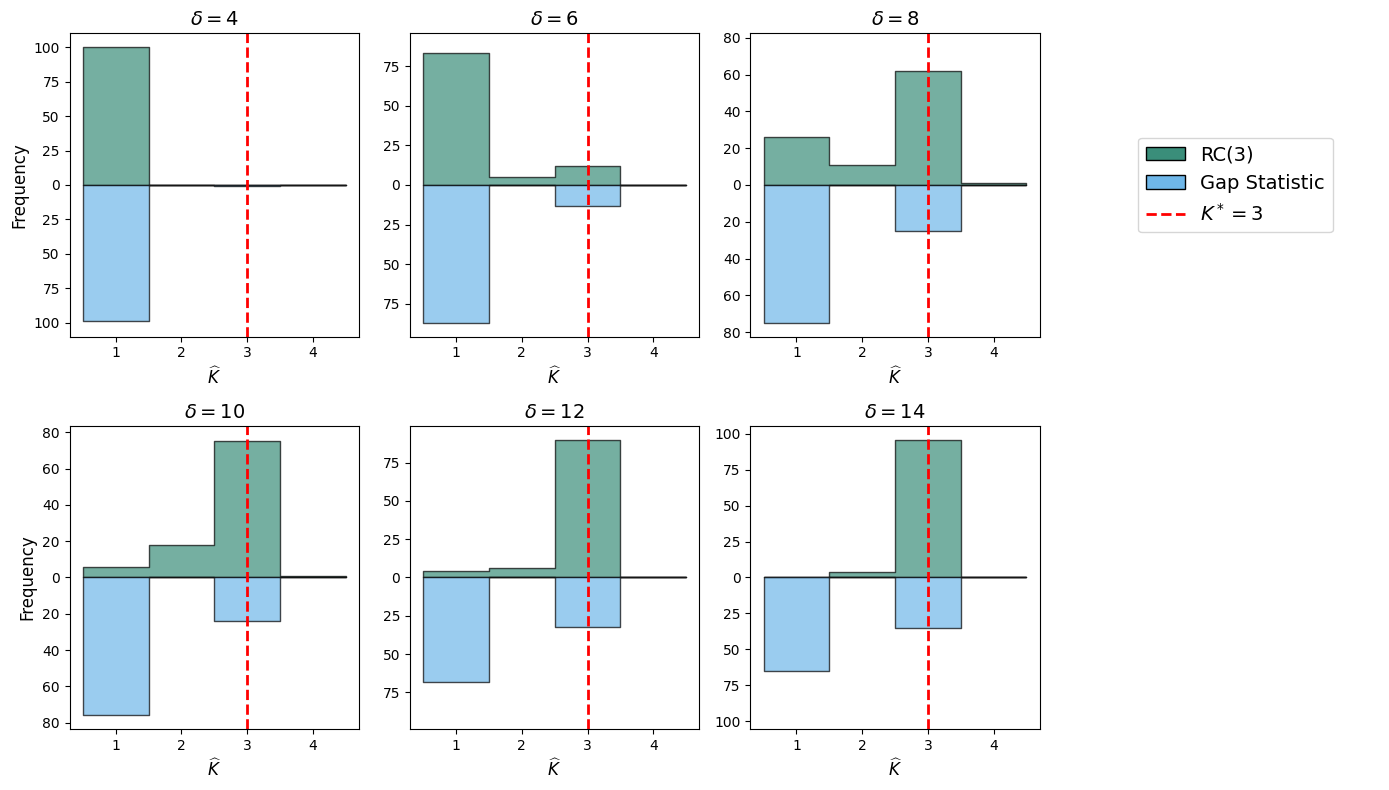}
    \caption{Paired histograms of the $\widehat{K}$ values selected by our proposed method and the gap statistic across varying values of $\delta$. As $\delta$ increases, our method consistently recovers the true number of clusters, while the gap statistic remains overly conservative, estimating $\widehat{K}=1$.}
    \label{fig:paired_histogram}
\end{figure}

Figure \ref{fig:paired_histogram} demonstrates that our sequential testing procedure more consistently recovers the correct clustering structure than the gap statistic across all settings. As cluster separation increases, it nearly always identifies the true $K$. By contrast, the gap statistic remains overly conservative and frequently underestimates the true number of clusters, even for well-separated clusters.
We also evaluate the performance of Algorithm~\ref{alg:alpha_spending_stop} under varying true numbers of clusters $K^*\in\{1,2,\dots,10\}$. Overall, the proposed method accurately recovers the true number of clusters when $K^*$ is small and reveals more clustering structure than the gap statistic as $K^*$ increases. The results of this experiment are provided in Appendix~\ref{subsec:varyK_addtional}.

\subsection{Real data analysis}
\label{sec:realdata}

In this section, we analyze the penguin dataset \citep{palmerpenguins} that was previously studied in \cite{gao2024selective,yun2023selective}. We focus on a subset of 107 female penguins observed during the years 2007–2008, using the three species \textit{Adelie}, \textit{Chinstrap}, and \textit{Gentoo} as ground-truth labels.
The clustering analysis is based solely on the two continuous measurements, bill length and flipper length.

First, we apply Algorithm \ref{alg:alpha_spending_stop} to estimate the number of clusters $K$, and then use the randomized hierarchical clustering algorithm in Algorithm \ref{alg:RAC} to assess the stability of the resulting clustering structure across repeated draws of randomization.
\begin{enumerate}
\item[(1):] To estimate the number of clusters $\widehat{K}$, we run $100$ independent repetitions of Algorithm \ref{alg:alpha_spending_stop} using complete linkage and a randomization level of $\tau^* = 0.10$, and select the mode of the resulting distribution as our estimate $\widehat{K}$. We employ the same $\alpha$-sequence, $n_{\min}$, and $n^*$ as described in Section \ref{sec:chooseK_experiments}.
\item[(2):] After selecting $\widehat{K}$, we repeatedly run our randomized clustering procedure $500$ times to assess stability of the clustering results. 
The pairwise co-occurrence matrix is constructed as follows: for every pair of samples $(i,j)$, we count the number of times they are assigned to the same cluster across all runs and divide by the total number of runs. Sharper color contrasts between blocks indicate more stable clustering results.
\end{enumerate}

\begin{figure}[h]
    \centering
    \includegraphics[width=0.9\linewidth, height=0.7\textheight, keepaspectratio]{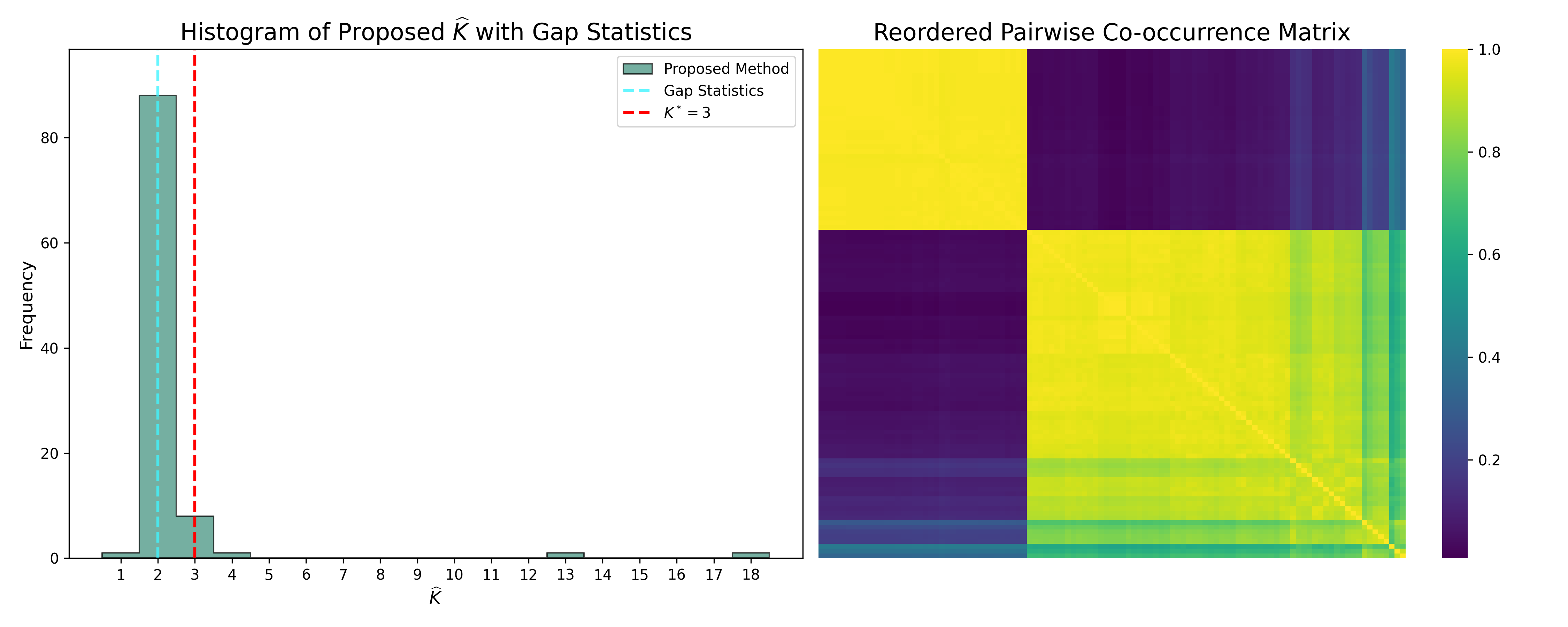}
    \caption{(Left) Histogram of the estimated number of clusters $\widehat{K}$. The blue dashed line shows $\widehat{K}=2$ given by the gap statistics. (Right) Heatmap displaying the pairwise co-occurrence matrix across 500 independent trials using $\widehat{K}=2$ and $\tau^*=0.10$.}
    \label{fig:penguin_hist}
\end{figure}

The left panel of Figure \ref{fig:penguin_hist} shows that both our proposed method and the gap statistic consistently select $\widehat{K} = 2$, where the red vertical line represents the true number of clusters, $K^* = 3$, and the blue vertical line indicates the number of clusters estimated by the gap statistic. The right panel displays the co-occurrence heatmap, which reveals two large, high-stability blocks. However, some uneven shading within the larger block suggests the presence of some within-cluster structure, consistent with the fact that the underlying data includes three species.
This observation aligns with the guarantee in Theorem \ref{thm:overest:control}, which ensures that our sequential testing procedure is unlikely to split true clusters, although it does not guarantee the recovery of all true clusters present in the data.

\section{Conclusion}
\label{sec:conclusion}

Validating clustering results is challenging because the same data cannot be used both to form the clusters and to test whether those clusters are genuine. 
In this paper, we address this challenge by developing tests that quantify the evidence against each merge in a hierarchical clustering algorithm. 
The resulting p-values are then used to determine when to stop the merging process and, in turn, to estimate the number of clusters. 

At the core of our approach is a novel randomized hierarchical clustering method. 
Randomization gives researchers flexible control to retain more information for inference or validation, without compromising the quality of selection.
In particular, the p-values produced by our testing procedure allow the same data to be used twice for both clustering and validation, resulting in more powerful tests while controlling the desired Type I error rate. 

Beyond the statistical advantages of using the same data for both clustering and validation, our proposed randomization scheme offers a practical strategy to construct tractable inference in hierarchical clustering.
It can be applied with any linkage without requiring the derivation of linkage-specific tests or the use of Monte Carlo approximations. 
This gives our approach an advantage over existing selective inferential methods, which can only explicitly characterize the conditioning set for certain linkages and dissimilarities based on Euclidean distances, but must rely on costly sampling-based approximations for others.

Finally, interpreting hierarchical clustering output relies on selecting an appropriate dendrogram level or cut height in order to extract and assess the resulting cluster structure. 
To this end, the p-values developed in our paper---with the appropriate conditional Type I error guarantees---are used to construct a sequential testing procedure that selects the number of clusters with a rigorous probabilistic guarantee against overestimation.

There are several promising directions for future work.
In this paper, we have derived exact p-values under normality.
To extend our approach to distribution-free settings, future work could investigate the construction of asymptotically valid p-values by using the transfer of conditional validity to other distributions or to non-parametric settings, similar to what has recently been done for randomized decision trees \citep{bakshi2025classification}.
Moreover, although our current approach focuses on hierarchical clustering algorithms, the potential of our randomization scheme to develop inference for other types of clustering methodology and for other types of cluster-based hypotheses remains to be explored.

\bibliography{references.bib}

\newpage
\appendix
\phantomsection\label{supplementary-material}
\bigskip

\begin{center}

{\large\bf APPENDIX}

\end{center}

\section{Proofs}
\subsection{Proof of Proposition \ref{prop:limit}} \label{proof:prop1}

We first present a preliminary result establishing that, at a given step $t$, the random merge made by our algorithm converges to the set of admissible merges in traditional clustering, conditional on the previous history of merges and the observed data.

\begin{lemma}
Suppose $(\tradmergecol{t-1}_o,\tradmerge{t}_o)\in \Omega_o^{*(t)}$ be a merge sequence from traditional hierarchical clustering applied to a fixed data matrix $X_o$.
Then, we have that
     \[
     \lim_{\tau^* \rightarrow 0} \bbP((\tradmergecol{t-1}_o,\merge{t})\in \Omega_o^{*(t)} |\mergecol{t-1} =\tradmergecol{t-1}_o,X=X_o ) =1.
     \] 
\label{lem:limit_t} 
\end{lemma}

\begin{proof}
By definition of traditional hierarchical clustering,
$$
d(\tradmerge{t}_o;X_o) < d(M; X_o),\quad \forall \quad M \notin \underset{M'\in \mathcal{M}(\tradmergecol{t-1}_o)}{\text{argmin}} d(M', X_o).
$$

Now, observe the following:
\begin{align*}
     \bbP(\merge{t}= M|\mergecol{t-1} =\tradmergecol{t-1}_o,X=X_o)  &= \dfrac{\exp \left(-\dfrac{d(M; X_o)}{\tau(\possiblemerges{t}_o,X_o)}\right)}{\displaystyle\sum_{M'\in \possiblemerges{t}_o}\exp \left(-\dfrac{d(M';X_o)}{\tau(\possiblemerges{t}_o,X_o)}\right)} \times \dfrac{\exp\left(\dfrac{d(\tradmerge{t}_o;X_o)}{\tau(\possiblemerges{t}_o,X_o)}\right)}{\exp\left(\dfrac{d(\tradmerge{t}_o;X_o)}{\tau(\possiblemerges{t}_o,X_o)}\right)}\\
        &=\dfrac{\exp\left(-\dfrac{1}{\tau(\possiblemerges{t}_o,X_o)} (d(M; X_o) - d(\tradmerge{t}_o;X_o))\right)}{\displaystyle\sum_{M'\in \possiblemerges{t}_o}\exp\left(-\dfrac{1}{\tau(\possiblemerges{t}_o,X_o)} (d(M'; X_o) - d(\tradmerge{t}_o;X_o))\right)}\\
        &= \frac{\exp\left(-\dfrac{|\possiblemerges{t}_o|}{\tau^*} \times \dfrac{(d(M; X_o) - d(\tradmerge{t}_o; X_o))}{\sum_{M'\in\possiblemerges{t}_o}d(M'; X_o)}\right)}{1+\displaystyle\sum_{M' \in \possiblemerges{t}_o\setminus \tradmerge{t}_o}\exp\left(-\dfrac{|\possiblemerges{t}_o|}{\tau^*} \times \dfrac{(d(M; X_o) - d(\tradmerge{t}_o; X_o))}{\sum_{M'\in\possiblemerges{t}_o}d(M'; X_o)}\right)}\\
        &< \exp\left(-\dfrac{|\possiblemerges{t}_o|}{\tau^*} \times \dfrac{(d(M; X_o) - d(\tradmerge{t}_o; X_o))}{\sum_{M'\in\possiblemerges{t}_o}d(M'; X_o)}\right).
\end{align*}
As $\tau^*\rightarrow 0$, the right-hand side of the above inequality tends to $0$ since $d( M;X_o) - d(\tradmerge{t}_o;X_o) > 0$. Hence, \[\bbP(\merge{t}= M|\mergecol{t-1} =\tradmergecol{t-1}_o,X=X_o) \rightarrow 0\] for all $M \notin \underset{M'\in \mathcal{M}(\tradmergecol{t-1}_o)}{\text{argmin}} d(M', X_o)$. Since the number of possible merges is finite, it follows that
\[
\lim_{\tau^* \to 0}
\bbP(\merge{t}\in \underset{M'\in \mathcal{M}(\tradmergecol{t-1}_o)}{\text{argmin}} d(M', X_o) \mid 
\mergecol{t-1}=\tradmergecol{t-1}_o, X=X_o)=1,
\]
which is equivalent to
$$
\lim_{\tau^* \rightarrow 0} \bbP((\tradmergecol{t-1}_o,\merge{t})\in \Omega_o^{*(t)} |\mergecol{t-1} =\tradmergecol{t-1}_o,X=X_o ) =1.
$$
\end{proof}

Now, we are ready to present to proof for Proposition \ref{prop:limit}.
\begin{proof}
We prove this result by induction. When $t=1$, we have that
\[\bbP(\mergecol{1}\in \Omega^{*(1)}_o| X=X_o) \rightarrow 1\]
as $\tau^* \rightarrow 0$, directly from Lemma~\ref{lem:limit_t}.

Now suppose that the claim holds up to step $t-1$, i.e.
\begin{equation}
\lim_{\tau^*\rightarrow 0 }\bbP(\mergecol{t-1} \in \Omega^{*(t-1)}_o|X=X_o)=1, 
    \label{eq:limit_t-1}
\end{equation}

At step $t$, by the law of total probability,
\begin{equation*}
    \begin{aligned}
    &\quad\bbP(\mergecol{t} \in \Omega^{*(t)}_o|X=X_o) \\
    &=
\sum_{\tradmergecol{t-1}_o\in \Omega_o^{*(t-1)}}
\bbP(
(\tradmergecol{t-1}_o,\merge t)\in \Omega_o^{*(t)},
\ \mergecol{t-1}=\tradmergecol{t-1}_o
\mid X=X_o
) \\
    &= \sum_{\tradmergecol{t-1}_o \in \Omega^{*(t-1)}_o} \bbP((\tradmergecol{t-1}_o,\merge{t})\in \Omega^{*(t)}_o \mid \mergecol{t-1} = \tradmergecol{t-1}_o, X=X_o)\times \bbP(\mergecol{t-1} = \tradmergecol{t-1}_o\mid X=X_o)
    \end{aligned}
\end{equation*}
By Lemma~\ref{lem:limit_t}, for each fixed $\tradmergecol{t-1}_o \in \Omega^{*(t-1)}_o$,
$$\lim_{\tau^*\rightarrow 0}\bbP((\tradmergecol{t-1}_o,\merge t)\in \Omega_o^{*(t)}
\mid \mergecol{t-1}=\tradmergecol{t-1}_o,\ X=X_o) = 1.$$
Therefore, it holds that
\begin{equation*}
\begin{aligned}
\lim_{\tau^* \rightarrow 0}\bbP(\mergecol{t} \in \Omega^{*(t)}_o|X=X_o)  &= \lim_{\tau^* \rightarrow 0}\sum_{\tradmergecol{t-1}_o\in \Omega^{*(t-1)}_o} \bbP(\mergecol{t-1} = \tradmergecol{t-1}_o\mid X=X_o)\\
&= \lim_{\tau^* \rightarrow 0}\bbP(\mergecol{t-1}\in \Omega^{*(t-1)}_o \mid X=X_o)=1,
\end{aligned}
\end{equation*}
where the second equality comes from the definition of $\Omega^{*(t-1)}_o$ in \eqref{eq:set_mergeseq} and the last equality comes from \eqref{eq:limit_t-1}. This completes the proof.
\end{proof}

\subsection{Proof of Lemma \ref{lem:recX}}\label{proof:lem2}

\begin{proof}

Note that the observed data matrix $X_o$ can be decomposed as follows:
\begin{align*}
    X_o &= B^{(t)}_oX_o + W^{(t)}_oX_o + (I_n-B^{(t)}_o-W^{(t)}_o)X_o\\
    &= \eta_o^{(t)} \|B^{(t)}_oX_o\|_F + \gamma_o^{(t)}\|W^{(t)}_oX_o\|_F + \Gamma_o^{(t)}\\
    &= \sqrt{\Delta_o^{(t)}}\cdot \left(\eta_o^{(t)} \sqrt{\frac{\|B^{(t)}_oX_o\|^2_F}{\|B^{(t)}_oX_o\|^2_F + \|W^{(t)}_oX_o\|^2_F}} +\gamma_o^{(t)} \sqrt{\frac{\|W^{(t)}_oX_o\|^2_F}{\|B^{(t)}_oX_o\|^2_F + \|W^{(t)}_oX_o\|^2_F}}\right) + \Gamma_o^{(t)}.
\end{align*}
The right-hand side of the above display further equals
\begin{align*}
    \sqrt{\Delta_o^{(t)}}\cdot \left(\eta_o^{(t)} \sqrt{\frac{R_o^{(t)}}{N_o^{(t)}-2+R_o^{(t)}}} +\gamma_o^{(t)} \sqrt{\frac{N_o^{(t)}-2}{N_o^{(t)}-2+R_o^{(t)}}}\right) + \Gamma_o^{(t)},
\end{align*}
by the definition of $R_o^{(t)} = (N_o^{(t)}-2)\cdot \dfrac{\|B^{(t)}_oX_o\|^2_F}{\|W^{(t)}_oX_o\|^2_F}$. 

By the definition of the function $X(r;\mathcal{A}^{(t)}_o)$, we see that upon setting $r = R_o^{(t)}$, its value coincides with the above expression for $X_o$, which establishes the claim. 
    
\end{proof}

\subsection{Proof of Lemma~\ref{lem:inf_target_dist}} \label{proof:lem:inf_target_dist}
\begin{proof}

First, we show that under the null hypothesis, the following independence between random variables holds:
\[R ^{(t)}\indep \Gamma^{(t)} \quad \text{and}\quad (\|B^{(t)}_o X\|_F,\|W^{(t)}_o X\|_F) \indep \left(\eta^{(t)},\gamma^{(t)}\right)\]

First, due to the isotropic covariance of $X$, we note that $B^{(t)}_oX$, $W^{(t)}_oX$, and $\Gamma^{(t)}=(I_n - B^{(t)}_o - W^{(t)}_o)X$ are mutually independent. 
Consequently, $R ^{(t)}\indep \Gamma^{(t)}$.

Moreover, under the null hypothesis, we have $\mathbb{E}[B^{(t)}_o X] = \mathbb{E}[W^{(t)}_o X] = 0$. 
By the properties of the Gaussian distribution with zero mean and isotropic covariance, it follows that $(\|B^{(t)}_o X\|_F,\|W^{(t)}_o X\|_F) \indep \left(\eta^{(t)},\gamma^{(t)}\right)$. Therefore, $R^{(t)}$ as a statistic that is a measurable function of $\|B^{(t)}_o X\|_F$, $\|W^{(t)}_o X\|_F$ remains independent of $\left(\eta^{(t)},\gamma^{(t)}\right)$.

Next, observe that the event $\mathcal{A}^{(t)} = \mathcal{A}^{(t)}_o$ can be written as 
\[\{\Delta^{(t)}=\Delta^{(t)}_o\}\cap \{\eta^{(t)}=\eta^{(t)}_o,\gamma^{(t)}=\gamma^{(t)}_o, \Gamma^{(t)} = \Gamma^{(t)}_o\},\]
where the second component involves only variables that are independent of $R^{(t)}$. We conclude that the distribution in our claim, 
$$
R^{(t)}  \Big\lvert \mathcal A^{(t)}=\mathcal A_o^{(t)}
$$ 
coincides with the distribution of $R^{(t)} \mid \Delta^{(t)} = \Delta^{(t)}_o$. 

Finally, since $\|B^{(t)}_o X\|^2_F$ and $\|W^{(t)}_o X\|^2_F$ are independent under the null $H_0^{(t)}$, with $\sigma^{-2} \cdot \|B^{(t)}_oX\|^2_F \sim \chi^2_p$ and $\sigma^{-2}\cdot\|W^{(t)}_o X\|^2_F \sim \chi^2_{(N_o^{(t)}-2)p}$, we have that $R^{(t)}$ is also independent of $\Delta^{(t)} = \|B^{(t)}_oX\|^2_F + \|W^{(t)}_o X\|^2_F$ and follows $F_{p,(N_o^{(t)}-2)p}$ distribution.  This follows from the properties of the $\chi^2$ and F distributions that for two independent random variables $Y\sim \chi^2_a$ and $Z\sim \chi^2_b$, $\frac{Y/a}{Z/b}$ is independent of their sum $Y+Z$, and follows a $F_{a,b}$ distribution. 

Since $R^{(t)}$, for fixed groups or clusters, is distributed as $F_{p,(N_o^{(t)}-2)p}$ under $H_0^{(t)}$, the claim follows.
    
\end{proof}

\subsection{Proof of Theorem~\ref{thm:formula_density}}\label{proof:thm1}
We first present a preliminary result that derives the probability of the observed merge sequence conditional on the test statistic and the auxiliary statistics.

\begin{lemma}\label{lem:sel_prob_product}
For $t\geq 1$, we have that
\[
\mathbb{P}\!\left(M^{(1)}=M_o^{(1)},\ldots,M^{(t)}=M_o^{(t)} \mid R^{(t)} = r,\; \mathcal{A}^{(t)} = \mathcal{A}^{(t)}_o\right)
= \prod_{s=1}^t \selprobtXu{s}{s}{r}
\]
where
$$\selprobtXu{s}{s}{r}= \mathbb{P}\left(M^{(s)}= M_o^{(s)} |\mergecol{s-1}=\mergecolobs{s-1},X= X(r; \mathcal{A}_o^{(t)})\right)$$
  is the sampling probability that $M^{(s)}= M_o^{(s)}$ given $X= X(r; \mathcal{A}_o^{(t)})$ and the merge history $\mergecol{s-1}=\mergecolobs{s-1}$, for $s\in \{1,2,\ldots, t\}$, as defined in \eqref{eq:sel_prob}.  
\end{lemma}
\begin{proof}
By Lemma~\ref{lem:recX}, we have that 
\begin{equation}
\begin{aligned}
&\bbP\left(M^{(1)}=M_o^{(1)},\ldots,M^{(t)}=M_o^{(t)}\; | \;  R^{(t)}=r, \mathcal A^{(t)}=\mathcal A_o^{(t)}\right)\\
&= \bbP\left(M^{(1)}=M_o^{(1)},\ldots,M^{(t)}=M_o^{(t)} \; \mid X = X(r; \mathcal{A}_o^{(t)})\right).
\end{aligned}
\label{correction:factor}
\end{equation}

By the chain rule of probability, one can further write the probability on the right-hand side of the above-stated equality as
\[\bbP(\merge{1} = \merge{1}_o\mid X=X(r; \mathcal{A}_o^{(t)}))\times\prod_{s=2}^t \bbP(\merge{s} = \mergeobs{s}\mid \mergecol{s-1} = \mergecolobs{s-1},X=X(r; \mathcal{A}_o^{(t)}))\]

From the definition of our sampling scheme in \eqref{eq:sel_prob}, it follows that the merge variable $M^{(s)}$, conditional on $\mergecol{s-1} = \mergecolobs{s-1}$ and $X = X(r; \mathcal{A}_o^{(t)})$, takes the value $M_o^{(s)}$ with probability $\selprobtXu{s}{s}{r}$, i.e.
$$\bbP(\merge{s} = \mergeobs{s}\mid \mergecol{s-1} = \mergecolobs{s-1},X=X(r; \mathcal{A}_o^{(t)}))=\selprobtXu{s}{s}{r}.$$
Hence, the probability in the claim factorizes as
\begin{align*}
&\mathbb P\!\left(\bigcap_{s=1}^t \{ \merge{s} = M_o^{(s)} \}\mid X = X(r; \mathcal{A}_o^{(t)})\right)=\prod_{s=1}^t \selprobtXu{s}{s}{r}.
\end{align*}
\end{proof}

Now, we are ready to present the proof for Theorem~\ref{thm:formula_density}.
\begin{proof}
Let the conditional density of 
$R^{(t)}  \Big\lvert \mathcal A^{(t)}=\mathcal A_o^{(t)}$
at $r$ be denoted by
$\ell_{R^{(t)}}(r \; \lvert \;  \mathcal A^{(t)}=\mathcal A_o^{(t)}).$
Then, using Bayes' rule, we observe that the density of the conditional distribution 
$R^{(t)}  \Big \lvert \{ \mergecol{t} = \mergecolobs{t}, \mathcal A^{(t)}=\mathcal A_o^{(t)} \},$
evaluated at $r$, is proportional to
\begin{equation}
\begin{aligned}
&\ell_{R^{(t)}}(r \; \lvert \;  \mathcal A^{(t)}=\mathcal A_o^{(t)}  )\cdot \bbP(\mergecol{t} = \mergecolobs{t}\; | \;  R^{(t)}=r,  \mathcal A^{(t)}=\mathcal A_o^{(t)} ).
\end{aligned}
\label{cond:density:corrected}
\end{equation}

Using Lemma \ref{lem:inf_target_dist}, under the null $H_0^{(t)}$, it follows that
\begin{equation}                        
\begin{aligned}
    \ell_{R^{(t)}}(r \; \lvert \;\mathcal A^{(t)}=\mathcal A_o^{(t)})&= \ell_{F_{p,(N_o^{(t)}-2)p}}(r).
\end{aligned}
\label{pre:selective:density}
\end{equation}
Furthermore, by Lemma~\ref{lem:sel_prob_product}, we have that 
\begin{equation}
\begin{aligned}
&\bbP\left(\mergecol{t} = \mergecolobs{t}\; | \;  R^{(t)}=r, \mathcal A^{(t)}=\mathcal A_o^{(t)}\right)\\
&= \bbP\left(\bigcap_{s=1}^t \{ \merge{s} = M_o^{(s)} \}\mid X = X(r; \mathcal{A}_o^{(t)})\right)= \prod_{s=1}^{t} \selprobtXu{s}{s}{r}.
\end{aligned}
\label{correction:factor}
\end{equation}
Combining \eqref{pre:selective:density} with \eqref{correction:factor}, we obtain that the expression in \eqref{cond:density:corrected} takes the value
$ \ell_{F_{p,(N_o^{(t)}-2)p}}(r) \times \prod_{s=1}^{t} \selprobtXu{s}{s}{r}.$
That is, the conditional density of interest is given by 
\[
\frac{ \ell_{F_{p,(N_o^{(t)}-2)p}}(r)\times \prod_{s=1}^{t} \selprobtXu{s}{s}{r}}{\int_0^\infty  \ell_{F_{p,(N_o^{(t)}-2)p}}(r') \times \prod_{s=1}^{t} \selprobtXu{s}{s}{r'}dr'}.
\]
As a result, the CDF of this distribution, evaluated at $r$, is as stated in the theorem.
\end{proof}

\subsection{Proof of Theorem \ref{thm:cond_p_value}}\label{proof:thm2}
\begin{proof}
Fixing some notation, let
\begin{equation*}
\begin{aligned}
\mathbb{P}_{\mergecolobs{t}}  (\mathcal{A}^{(t)}_o)
&=\notag\bbP_{H_0^{(t)}} \left(\pvalue(R^{(t)};\mathcal{A}^{(t)})\leq \alpha|\mergecol{t}=\mergecolobs{t},\mathcal A^{(t)}=\mathcal A_o^{(t)}\right)\\
&= \mathbb{E}_{H_0^{(t)}}\left[\mathbbm{1}\{\pvalue(R^{(t)};\mathcal{A}^{(t)})\leq \alpha\}|\mergecol{t}=\mergecolobs{t},\mathcal A^{(t)}=\mathcal A_o^{(t)}\right].
\end{aligned}
\end{equation*}

Then, by the probability integral transform, we obtain that
\begin{align*}
       \mathbb{P}_{\mergecolobs{t}}  (\mathcal{A}^{(t)}_o)=\alpha.
\end{align*}
Finally, we apply the law of iterated expectation to conclude that
\begin{align*}
    &\bbP_{H_0^{(t)}}\left(\pvalue(R^{(t)};\mathcal{A}^{(t)})\leq\alpha|\mergecol{t} = \mergecolobs{t}\right)=\mathbb{E}_{H_0^{(t)}}\left[\mathbb{P}_{\mergecolobs{t}}  (\mathcal{A}^{(t)}_o) |\mergecol{t}=\mergecolobs{t}\right]= \alpha.
\end{align*}
\end{proof}

\subsection{Proof of Proposition \ref{prop:pvals_conv}}\label{proof:prop2}
\begin{proof}
By Lemma \ref{lem:limit_t}, for each step $1\leq s\leq t$, it holds almost everywhere that, as $\tau^* \rightarrow 0$,
\[\selprobtXu{s}{s}{r} \rightarrow  \mathbbm{1}\{\mergeobs{s} \in \Psi^{*(s)}(r)\},\]
where $\Psi^{*(s)}(r) = \displaystyle\argmin_{M\in \possiblemerges{s}_o}d(M;X(r;\mathcal{A}^{(s)}_o))$.

Hence, the weight function
\[w_{\tau^*}(r)  := \prod_{s=1}^t \selprobtXu{s}{s}{r}\]
converges almost everywhere to $\prod_{s=1}^{t}\mathbbm{1}\{\mergeobs{s} \in \Psi^{*(s)}(r)\}$.
Furthermore, the weight functions are uniformly bounded, i.e.,
$0\leq w_{\tau^*}(r) \leq 1$.
Since the integrands in the numerators and denominators of $\mathbb{F}^{(t)}(R^{(t)}_o;\mathcal{A}^{(t)}_o, \mergecolobs{t})$, derived in Theorem \ref{thm:formula_density}, converge almost everywhere, respectively, to the integrands of \eqref{eqn:hard:trunc}, applying the dominated convergence theorem to the numerator and denominator of $\mathbb{F}^{(t)}(R^{(t)}_o;\mathcal{A}^{(t)}_o, \mergecolobs{t})$ yields the stated limit. 
\end{proof}

\subsection{Proof of Theorem \ref{thm:overest:control}} \label{proof:thm3}
To prove Theorem \ref{thm:overest:control}, we define $\mathcal{V} = \{\text{a true null } H_0^{(t)} \text{ is rejected}\} \cap \mathcal{T}$, the event that we falsely detect a difference between two clusters with the same means, and thus erroneously fail to merge them in the sequential procedure. We first present a supporting result showing that Algorithm \ref{alg:alpha_spending_stop} controls the family-wise error rate (FWER), which equals $\bbP(\mathcal{V})$.

\begin{lemma}
Suppose Algorithm $\ref{alg:alpha_spending_stop}$ is run with total significance level $\alpha$ and a pre-specified sequence $\{\alpha_1,\alpha_2,\dots,\alpha_{n-1}\}$ satisfying $\sum^{n-1}_{j=1} \alpha_j = \alpha$. 
Then, our sequential testing procedure, applied along the merges of the hierarchical clustering algorithm and stopping at the first rejection, controls the family-wise error rate at level $\alpha$, i.e.,
\[\bbP(\mathcal{V})\leq \alpha.\]
\label{lem:FWER}
\end{lemma}

\begin{proof}
For each step $t$, consider the event
\begin{equation*}
    \begin{gathered}
     \mathcal{R}^{(t)} =\{\text{no rejection before step $t$} \}\cap \{H_0^{(t)} \text{ is true}\} \cap \{\widehat{p}^{(t)}\leq \alpha^{(t)}\},
    \end{gathered}
\end{equation*}
where $\alpha^{(t)}:=\alpha^{(t)}(\mergecol{t})$ is the assigned sequence in Algorithm \ref{alg:alpha_spending_stop}, that is allowed to depend on the clusters merged up to step $t$, particularly on the sizes of the clusters merged during step $t$.

Observe that
\begin{align*}
    \bbP(\mathcal{V}) = \bbP\left(\bigcup_{t=1}^{t^*} \mathcal{R}^{(t)}\right) &= \sum_{t=1}^{t^*}\bbP(\mathcal{R}^{(t)}) \\
    &\leq \sum_{t=1}^{t^*} \mathbb{E}\left[  \bbP_{H_0^{(t)}}\left(\widehat{p}^{(t)}\leq \alpha^{(t)}(\mergecol{t}) | \mergecol{t}=\mergecolobs{t}\right)\right] \\
    &= \sum_{t=1}^{t^*}\mathbb{E}\left[ \alpha^{(t)}(\mergecolobs{t})\right]\\
    &= \mathbb{E}\left[ \sum_{t=1}^{t^*}\alpha^{(t)}(\mergecolobs{t})\right]
    = \sum_{t=1}^{t^*}\alpha^{(t)}(\mergecolobs{t}) \leq \alpha.
\end{align*}
Here, the first equality holds because Algorithm \ref{alg:alpha_spending_stop} stops at the first rejection, which means the events $\mathcal{R}^{(t)}$ are disjoint.
The second inequality is satisfied trivially because $\{H_0^{(t)} \text{ is true}\} \cap \{\widehat{p}^{(t)}\leq \alpha^{(t)}(\mergecol{t})\} \subseteq \mathcal{R}^{(t)}$.
The third equality is due to the conditional validity of the p-value $\widehat{p}^{(t)}$, as established in Theorem \ref{thm:cond_p_value}, and once we condition on the sequence of merges, the adaptively-determined level of significance $\alpha^{(t)}(\mergecolobs{t})$ can be treated as fixed, i.e.,
\[\bbP(\widehat{p}^{(t)}\leq \alpha^{(t)}(\mergecol{t})|\mergecol{t}=\mergecolobs{t}) = \alpha^{(t)}(\mergecolobs{t}).
\]

\end{proof}

We are now ready to present the proof of Theorem \ref{thm:overest:control} using Lemma \ref{lem:FWER}.

\begin{proof}[Theorem \ref{thm:overest:control}]
The proof of this result follows from observing that
$$
\bbP\left(\{\widehat{K}> K^*\} \cap \mathcal{T} \right) \leq \bbP(\mathcal{V}).
$$
By Lemma \ref{lem:FWER}, we have $\bbP(\mathcal{V}) \leq \alpha$, which holds due to the construction of our $\alpha$-spending procedure.
Therefore, the claim follows. 
\end{proof}

\section{Extension: constructing p-value under non-spherical covariance}
\label{subsec:non_spherical}

The conditional p-values developed in Section \ref{sec:pvals} are derived under the model in \eqref{eq:model}, where $Cov(X_i) = \sigma^2 I_p$ for an unknown $\sigma$. In this section, we describe our construction of p-values under the null hypothesis for the test considered by \cite{gao2024selective} with a known non-spherical covariance matrix $\Sigma$, using our randomized approach. 
Specifically, we consider the model 
\begin{equation}
       X \sim \mathcal{MN}_{n\times p}(\mu,I_n,\Sigma),
\end{equation}
where $\Sigma\in \mathbb{R}^{p\times p}$ is a known positive definite matrix. 
This includes the special case where $\Sigma = \sigma^2 I_p$ for known $\sigma^2$.

\begin{remark}
In practice, $\Sigma$ is not known and must be estimated from the observed data.
Formally, we expect that an asymptotic justification can show that this class of tests extends to consistent plug-in estimates of $\Sigma$, thereby allowing their use even when $\Sigma$ is estimated. 
For an example of this type of justification, see \cite{bakshi2025classification}.
A detailed investigation of this is left to future work.
\end{remark}

For the observed realization of the data matrix $X  = X_o$ and an observed merge $\{\merge{t}=\mergeobs{t}\}$, where $\mergeobs{t} = \left(\wincluobs{t}{1},\wincluobs{t}{2}\right)$, we consider the null hypothesis in \cite{gao2024selective}, defined as
\begin{equation}
\label{eq:hypothesis_chi}
    H_0^{(t)}: \bar{\mu}_{\wincluobs{t}{1}}=\bar{\mu}_{\wincluobs{t}{2}},
\end{equation}
where $\bar{\mu}_C = \frac{1}{|C|}\sum_{i\in C}\mu_i$ is defined as the mean of the population means in cluster $C$.

\subsection{Test statistics and preliminaries}

As done earlier, we first define the test statistic and consider its distribution conditional on appropriately identified auxiliary statistics under the null hypothesis, assuming the clusters are predefined.

To test \eqref{eq:hypothesis_chi}, we consider the test statistic
\begin{equation}\label{eq:test_stat_chi}
    U^{(t)}\left(X;\wincluobs{t}{1},\wincluobs{t}{2}\right) := \big\|\Sigma^{-\frac{1}{2}}X^T\nu_{\mergeobs{t}}\big\|_2,
\end{equation}
where $\nu_{\mergeobs{t}}$ is defined in \eqref{eq:nuM}. This statistic quantifies the separation between the means of the merged clusters, but in terms of the whitened data matrix $X$. 
For simplicity, we will denote $U^{(t)}\left(X;\wincluobs{t}{1},\wincluobs{t}{2}\right)$ simply as $U^{(t)}$ hereafter.

Define the auxiliary statistics for this test by 
$\mathcal{A}^{(t)}_{\Sigma} = \{\xi^{(t)},\pi^{(t)}\}$,
where 
\[\xi^{(t)} = dir(\Sigma^{-\frac{1}{2}}X^T\nu_{\mergeobs{t}})^T, \quad\pi^{(t)} = \left(I_n - \dfrac{\nu_{\mergeobs{t}}\nu_{\mergeobs{t}}^T}{\|\nu_{\mergeobs{t}}\|^2_2}\right)X.\]
Denote the observed value of 
$\mathcal{A}^{(t)}_{\Sigma}$ by 
\begin{equation}\label{eq:auxi_stats_chi}
    \mathcal{A}_{\Sigma,o}^{(t)}=  \{\xi_o^{(t)},\pi_o^{(t)}\}.
\end{equation}

Lemma \ref{lem:decompXSigma} provides a decomposition of $X$ in terms of the test statistic $U^{(t)}$ and the auxiliary statistics $\mathcal{A}^{(t)}_{\Sigma}$.
\begin{lemma}
For some $u \in \mathbb{R}^+$, defining the function 
\[X\left(u;\mathcal{A}^{(t)}_{\Sigma,o}\right) = u\dfrac{\nu_{\mergeobs{t}}}{\|\nu_{\mergeobs{t}}\|^2_2} \xi_o^{(t)}\Sigma^{\frac{1}{2}} + \pi_o^{(t)},\]
where $\nu_{\mergeobs{t}}$ is defined in \eqref{eq:nuM} and $\mathcal{A}^{(t)}_{\Sigma,o}$ is defined in \eqref{eq:auxi_stats_chi}.
It holds that
\[X(U^{(t)}_o;\mathcal{A}^{(t)}_{\Sigma,o}) = X_o,\]
where $U^{(t)}_o$ is the observed value of the test statistics defined in \eqref{eq:test_stat_chi} and $X_o$ is the observed value of the data matrix $X$.
\label{lem:decompXSigma}
\end{lemma}
\begin{proof}
Note that the observed data matrix $X_o$ can be decomposed as follows:
\begin{align*}
    X_o &=  \frac{\nu_{\mergeobs{t}}}{\|\nu_{\mergeobs{t}}\|^2_2}\nu_{\mergeobs{t}}^TX_o +\left(I_n - \frac{\nu_{\mergeobs{t}}\nu_{\mergeobs{t}}^T}{\|\nu_{\mergeobs{t}}\|_2^2}\right)X_o\\
    & = \|\Sigma^{-\frac{1}{2}}X_o^T\nu_{\mergeobs{t}}\|_2 \frac{\nu_{\mergeobs{t}}}{\|\nu_{\mergeobs{t}}\|^2_2}dir(\Sigma^{-\frac{1}{2}}X_o^T\nu_{\mergeobs{t}})^T\Sigma^{\frac{1}{2}}  +\left(I_n - \frac{\nu_{\mergeobs{t}}\nu_{\mergeobs{t}}^T}{\|\nu_{\mergeobs{t}}\|_2^2}\right)X_o \\
    & = U^{(t)}_o\frac{\nu_{\mergeobs{t}}}{\|\nu_{\mergeobs{t}}\|^2_2}\xi^{(t)}_o\Sigma^{\frac{1}{2}}+\pi_o^{(t)},
\end{align*}
where the second equality follows from expressing the projection term $\frac{\nu_{\mergeobs{t}}}{\|\nu_{\mergeobs{t}}\|^2_2}\nu_{\mergeobs{t}}^TX_o$ along $dir(\Sigma^{-\frac{1}{2}}X^T_o\nu_{\mergeobs{t}})$, and the last equality follows from the definition of $U^{(t)}_o$ defined is \eqref{eq:test_stat_chi} and $\mathcal{A}^{(t)}_{\Sigma,o}$ defined in \eqref{eq:auxi_stats_chi}.

By the definition of the function $X(u;\mathcal{A}^{(t)}_{\Sigma})$, we see that upon setting $u = U_o^{(t)}$, its value coincides with the above expression for $X_o$, which establishes the claim. 
\end{proof}
Assume that the merged clusters $\mergeobs{t} = \left(\wincluobs{t}{1}, \wincluobs{t}{2}\right)$ are predefined.
Then, in analogy with Lemma \ref{lem:inf_target_dist}, the following result establishes that, under the null in \eqref{eq:hypothesis_chi}, the test statistic $U^{(t)}$ follows a $\chi_p$ distribution when conditioned on the auxiliary statistics.

\begin{lemma}\label{lem:inf_target_dist_chi}
    Consider the null hypothesis in \eqref{eq:hypothesis_chi}, assuming that the clusters $(\wincluobs{t}{1},\wincluobs{t}{2})$ are predefined. 
    Then, 
    $$
    U^{(t)}|\mathcal{A}^{(t)}_{\Sigma}= \mathcal{A}_{\Sigma,o}^{(t)}\sim\|\nu_{\mergeobs{t}}\|_2\cdot\chi_p.
    $$
\end{lemma}

\begin{proof}
First, we show that under the null hypothesis $H_0^{(t)}: \bar{\mu}_{\wincluobs{t}{1}} = \bar{\mu}_{\wincluobs{t}{2}}$, it holds that
\[U^{(t)}\indep \pi^{(t)}, \text{ and } U^{(t)}\indep \xi^{(t)}.\]

To this end, since $\left(I_n - \dfrac{\nu_{\mergeobs{t}}\nu_{\mergeobs{t}}^T}{|\nu_{\mergeobs{t}}|^2_2}\right)$ is an orthogonal projection matrix onto the subspace orthogonal to $\nu_{\mergeobs{t}}$, it follows that $U^{(t)} \indep \pi^{(t)}$.

To show independence between $U^{(t)}$ and $ \xi^{(t)}$, observe that $X\Sigma^{-\frac{1}{2}}\sim \mathcal{MN}(\mu\Sigma^{-\frac{1}{2}},I_n, I_p)$, and therefore, 
\[\frac{\Sigma^{-\frac{1}{2}}X^T\nu_{\mergeobs{t}}}{\|\nu_{\mergeobs{t}}\|_2}\sim \mathcal{N}(0,I_p).\]
By the standard property of the isotropic normal distribution that the direction and magnitude are independent, we conclude that $U^{(t)} \indep \xi^{(t)}$.

Combining both independence results, we conclude that the conditional distribution $U^{(t)}|\mathcal{A}^{(t)}_{\Sigma}=\mathcal{A}^{(t)}_{\Sigma,o}$ is equal to the unconditional distribution of the random variable $U^{(t)}$. Lastly, under the null, it is straightforward to verify that \[\|\Sigma^{-\frac{1}{2}}X^T\nu_{\mergeobs{t}}\|^2_2\sim \|\nu_{\mergeobs{t}}\|^2_2\cdot\chi^2_p,\] which yields the desired result.
\end{proof}

\subsection{Computing the p-value}

To compute the conditional p-value for the null in \eqref{eq:hypothesis_chi}, we characterize the conditional distribution of 
\begin{equation}\label{eq:target_U}
    U^{(t)}\Big\lvert \left\{\mergecol{t}=\mergecolobs{t},\mathcal{A}_{\Sigma}=\mathcal{A}^{(t)}_{\Sigma,o}\right\}
\end{equation}

\begin{theorem}\label{thm:formula_density_chi}
    Under $H^{(t)}_0$ in \eqref{eq:hypothesis_chi}, the cumulative distribution function of the conditional distribution of the test statistics $U^{(t)}$ in \eqref{eq:target_U}, evaluated at $u$, is given by 
    \[\mathbb{F}^{(t)}(u;\mathcal{A}^{(t)}_{\Sigma,o},\mergecolobs{t}) =
    \frac{\int_0^{u}\ell_{\|\nu_{\mergeobs{t}}\|_2\cdot \chi_p}(u')\times \prod_{s=1}^{t} \selprobtXuchi{s}{s}{u'}du'}{\int_0^\infty\ell_{\|\nu_{\mergeobs{t}}\|_2\cdot \chi_p}(u')\times \prod_{s=1}^{t}\selprobtXuchi{s}{s}{u'}du'},\]
    where $\ell_{\|\nu_{\mergeobs{t}}\|_2\cdot \chi_p}(u)$ denotes the density of a $\|\nu_{\mergeobs{t}}\|_2 \cdot \chi_p$ random variable and $\selprobtXuchi{s}{s}{u'}$ is the probability that $M^{(s)}= M_o^{(s)}$ given the merge history $\mergecol{s-1} = \mergecolobs{s-1} $ and $X= X(u'; \mathcal{A}_o^{(t)})$, for $s\in \{1,2,\ldots, t\}$, as defined in \eqref{eq:sel_prob}.
\end{theorem}

\begin{proof}
    Denote the conditional density of 
    \[U^{(t)}|\mathcal{A}^{(t)}_{\Sigma} =\mathcal{A}^{(t)}_{\Sigma,o}\] evaluated at $u$ by $\ell_{U^{(t)}}\left(u|\mathcal{A}^{(t)}_{\Sigma} =\mathcal{A}^{(t)}_{\Sigma,o}\right)$. Then using Bayes' rule, we observe that the density of the conditional distribution 
    $U^{(t)}\big\lvert \{\mergecol{t}=\mergecolobs{t}, \mathcal{A}^{(t)}_{\Sigma} =\mathcal{A}^{(t)}_{\Sigma,o}\}$
    at $u$, is proportional to 
\begin{equation}
\ell_{U^{(t)}}(u \; \lvert \;\mathcal{A}^{(t)}_{\Sigma} =\mathcal{A}^{(t)}_{\Sigma,o})\times \bbP(\mergecol{t} = \mergecolobs{t}\; | \;  U^{(t)} = u,\mathcal{A}^{(t)}_{\Sigma} =\mathcal{A}^{(t)}_{\Sigma,o}).
\label{cond:density:corrected_chi}
\end{equation}

Using Lemma \ref{lem:inf_target_dist_chi}, under the null $H_0^{(t)}$, it follows that
\begin{equation}                        
\begin{aligned}
    \ell_{U^{(t)}}\left(u \; \lvert \;\mathcal{A}^{(t)}_\Sigma = \mathcal{A}^{(t)}_{\Sigma,o}\right)&= \ell_{\|\nu_{\mergeobs{t}}\|_2\cdot \chi_p}(u).
\end{aligned}
\label{pre:selective:density_chi}
\end{equation}
Furthermore, by Lemma \ref{lem:sel_prob_product}, we have that
\begin{equation}
\begin{aligned}
&\bbP\left(\mergecol{t} = \mergecolobs{t}\; | \; U^{(t)}=u, \mathcal{A}^{(t)}_\Sigma = \mathcal{A}^{(t)}_{\Sigma,o}\right)\\
&= \bbP\left(\bigcap_{s=1}^t\{\merge{s} = \mergeobs{s}\}|X = X(u; \mathcal{A}_{\Sigma,o}^{(t)})\right)= \prod_{s=1}^{t} \selprobtXuchi{s}{s}{u}.
\end{aligned}
\label{correction:factor_chi}
\end{equation}
Combining \eqref{pre:selective:density_chi} with \eqref{correction:factor_chi}, we obtain that the expression in \eqref{cond:density:corrected_chi} takes the value
$$ \ell_{\|\nu_{\mergeobs{t}}\|_2\cdot \chi_p}(u) \times \prod_{s=1}^{t} \selprobtXuchi{s}{s}{u}.$$
That is, the conditional density of interest is given by 
\[
\frac{ \ell_{\|\nu_{\mergeobs{t}}\|_2\cdot \chi_p}(u)\times \prod_{s=1}^{t} \selprobtXuchi{s}{s}{u}}{\int_0^\infty  \ell_{\|\nu_{\mergeobs{t}}\|_2\cdot \chi_p}(u') \times \prod_{s=1}^{t} \selprobtXuchi{s}{s}{u'}du'}.
\]
As a result, the CDF of this distribution, evaluated at $u$, is as stated in the theorem.
\end{proof}

We then define the conditional p-value for this test by $\pvalue_\Sigma =\pvalue_\Sigma(U^{(t)};\mathcal{A}^{(t)}_{\Sigma})$, which is a function of the test statistics and the auxiliary statistics. Let $\pvalue_{\Sigma,o} = \pvalue_{\Sigma}(U^{(t)}_o;\mathcal{A}^{(t)}_{\Sigma,o})$, where 
\[\pvalue_{\Sigma,o} = \mathbb{P}_{H_0^{(t)}}\left(U^{(t)} \geq U^{(t)}_o\big\lvert \mergecol{t} = \mergecolobs{t}, \mathcal{A}^{(t)}_\Sigma = \mathcal{A}^{(t)}_{\Sigma,o} \right),\]
denotes its value computed on the observed data $X_o$, where $U^{(t)}_o$ is the observed value of the test statistics $U^{(t)}$.
\begin{corollary}
It holds that $\pvalue_{\Sigma,o} = 1-\mathbb{F}^{(t)}(U^{(t)}_o;\mathcal{A}^{(t)}_{\Sigma,o},\mergecolobs{t})$.
\label{cor:pval_chi}
\end{corollary}

\begin{proof}
 This claim follows directly from the definition of $\pvalue_{\Sigma,o}$ and the conditional distribution derived in Theorem \ref{thm:formula_density_chi}.
\end{proof}

As shown earlier in Section \ref{sec:pvals}, Theorem \ref{thm:cond_p_value_chi} establishes that the test based on this p-value controls the conditional Type I error at the nominal level $\alpha$.

\begin{theorem}\label{thm:cond_p_value_chi}
We have that
\begin{equation*}
    \bbP_{H_0^{(t)}}\left(\pvalue_\Sigma(U^{(t)};\mathcal{A}_{\Sigma}^{(t)})\leq \alpha \; | \;\mergecol{t} =\mergecolobs{t}\right) = \alpha, \quad \text{ for any} \quad 0\leq\alpha\leq 1.
\end{equation*}
\end{theorem}

The proof of this theorem is identical to the proof of Theorem \ref{thm:cond_p_value} and is omitted from the paper.

\section{Additional simulation results}

\subsection{Clustering quality of randomized hierarchical clustering algorithm}\label{subsec:clustering_quality_addition}

\subsubsection{Definition of the clustering quality metrics}
\label{subsubsec:clustering_metrics}
Before presenting additional results on clustering quality, we first provide the definitions of the clustering quality metrics used in our study.

\begin{enumerate}
\item[(i)] The ratio between the within-cluster sum of squares (WCSS) and the total sum of squares (TSS), defined as 
\[
\frac{\text{WCSS}}{\text{TSS}} = \frac{\sum_{k=1}^K \sum_{i\in C_k}\|X_i - \Bar{X}_{C_k}\|^2_2}{\sum_{i=1}^n \|X_i - \Bar{X}\|^2_2},
\]
with $\Bar{X}_{C_k} = \dfrac{1}{|C_k|}\sum_{i\in C_k}X_i$ denoting the mean of cluster $k$, and $\Bar{X} = \dfrac{1}{n}\sum_{i=1}^n X_i$ representing the overall sample mean. This ratio measures the proportion of variability within clusters relative to the total variability, ranging from $0$ to $1$, with smaller values indicating more well-defined and compact clusters.

\item[(ii)] The Adjusted Rand Index (ARI), which measures the agreement between the estimated clustering $\widehat{C} = \{\widehat{C}_1, \dots, \widehat{C}_K\}$ and the true labels $C = \{C_1, \dots, C_{K'}\}$ of $n$ observations.
Let $n_{k,k'} = |\widehat{C}_k \cap C_{k'}|$ denote the number of points assigned to the estimated cluster $\hat{C}_k$ and the true cluster $C_{k'}$. 
The ARI is then defined as 
\[\text{ARI}= \dfrac{\sum_{k,k'}\binom{n_{k,k'}}{2} - \left[\sum_k \binom{a_k}{2}\sum_{k'}\binom{b_{k'}}{2}\right]\Big/\binom{n}{2}}{\dfrac{1}{2}\left[\sum_k \binom{a_k}{2}+ \sum_{k'}\binom{b_{k'}}{2}\right]-\left[\sum_k \binom{a_k}{2}\sum_{k'}\binom{b_{k'}}{2}\right]\Big/\binom{n}{2}},\]
where $a_k = \sum_{k'} n_{kk'}$ and $b_k' = \sum_k n_{kk'}$ denote, respectively, the total number of samples assigned to the estimated cluster $k$ and the total number of samples belonging to true class $k'$.
ARI ranges from $-1$ to $1$, with $1$ indicating perfect agreement between the estimated and true clustering, $0$ corresponding to random labelings and $-1$ indicating extreme disagreement that is worse than random labelings.
\end{enumerate}

\subsubsection{Additional experiments on evaluating clustering quality}\label{subsubsec:clustering_qual_add}
Continuing the simulations in Section~\ref{subsec:clustering_quality}, we now present additional results assessing the clustering performance of Algorithm~\ref{alg:RAC}. Specifically, in addition to comparing performance across different randomization level, we compare the performance of the randomized hierarchical clustering algorithm under varying levels of signal strength $\delta$, while fixing the randomization level at $\tau^* = 0.10$ (i.e., implementing RC(3)). 
Figure \ref{fig:clustering_quality_delta} shows that the randomized hierarchical clustering algorithm with $\tau^* = 0.10$ achieves performance very similar to the deterministic version across a wide range of signal strengths $\delta$. 
When the clusters are well separated (larger $\delta$), the randomized method attains nearly perfect recovery of the true clustering, as reflected by low WCSS/TSS ratios and high ARI values. 
Overall, these results suggest that a small amount of randomization does not degrade clustering quality, even as the signal strength varies. 
The benefits of the added randomization, however, become evident in the next sections, where we demonstrate the performance of our testing procedures for validating the clusters discovered by Algorithm \ref{alg:RAC}.

\begin{figure}[h]
        \centering
        \includegraphics[width = \linewidth]{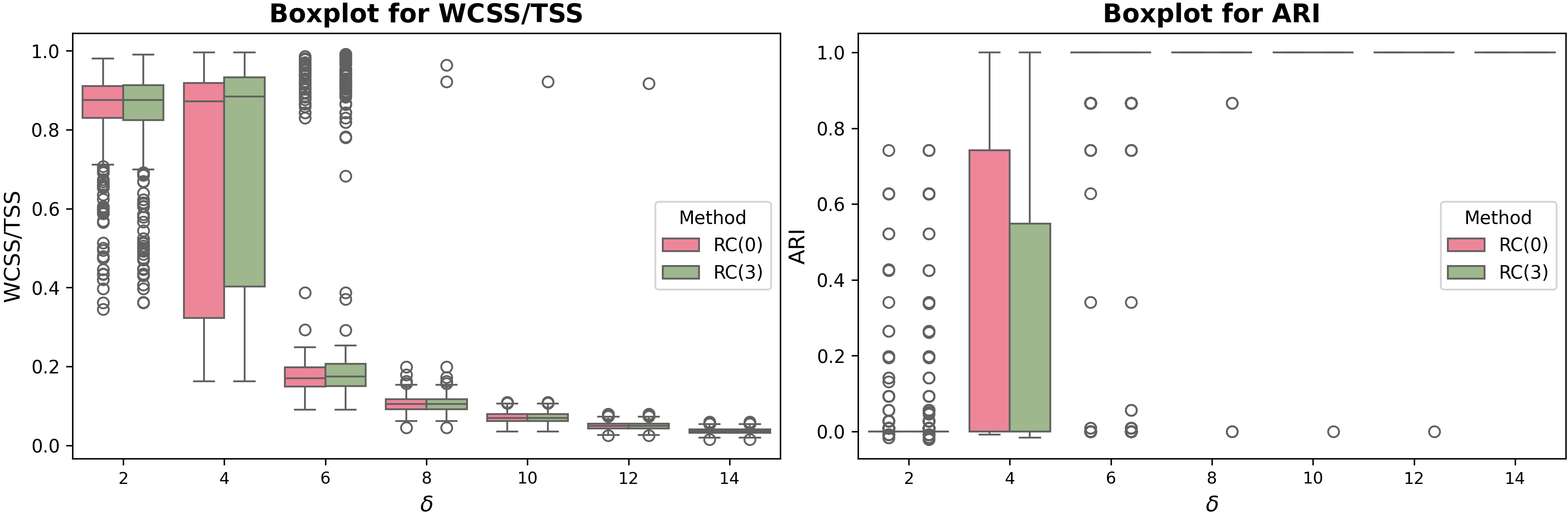}
    \caption{Comparison of clustering performance between deterministic and randomized hierarchical clustering across signal strengths $\delta$. (Left): Side-by-side boxplots of the WCSS/TSS showing how cluster compactness changes with $\delta$. (Right): Side-by-side boxplots of the Adjusted Rand Index (ARI), measuring agreement with the true clustering, which increases as $\delta$ increases.}
    \label{fig:clustering_quality_delta}
\end{figure}

\subsection{Additional experiments on evaluating p-value validity and power}\label{subsec:valdity_additional}

We begin this section by presenting additional experiment results for the experiments in Section \ref{sec:validitypower}. 

Following the setting described in Section \ref{sec:validitypower}, generate data under the null hypothesis and partition data into $K=3$ clusters. Figure \ref{fig:validity_random_k3} shows that the proposed method produces uniform p-values under the null hypothesis and successfully controls the Type I error rate. 

To evaluate power of the proposed p-value using $K=3$, we use the data generated from three equidistant clusters
\begin{equation}\label{eq:3clutsters}
    \mu_1 = \dots = \mu_{n/3} = \begin{bmatrix}
        0\\0
    \end{bmatrix}, \mu_{n/3+1}=\dots= \mu_{2n/3} = \begin{bmatrix}
        \delta\\0
    \end{bmatrix}, \mu_{2n/3+1}=\dots= \mu_{n} = \begin{bmatrix}
        \delta/2\\\sqrt{3}\delta/2
    \end{bmatrix}.
\end{equation}
Figure \ref{fig:power_k3} illustrates that the proposed method consistently achieves greater power compared to existing selective inference procedures.

\begin{figure}[h]
  \centering
  \includegraphics[width = \textwidth]{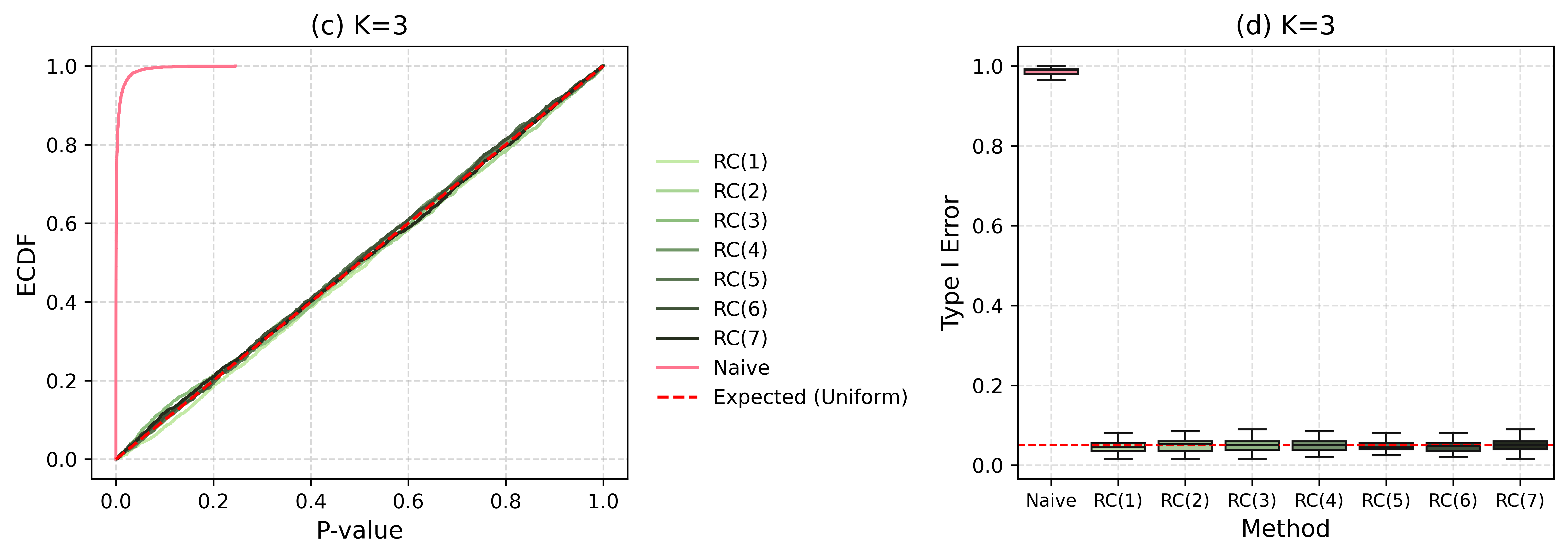}

  \caption{Comparison of the p-value ECDFs and Type I error rates simulated under the null hypothesis. (Left) ECDF plots for the proposed and baseline methods. (Right) Boxplots of the Type I error rates across different methods.}
  \label{fig:validity_random_k3}
\end{figure}

\begin{figure}[h]
    \centering
    \includegraphics[width=\linewidth]{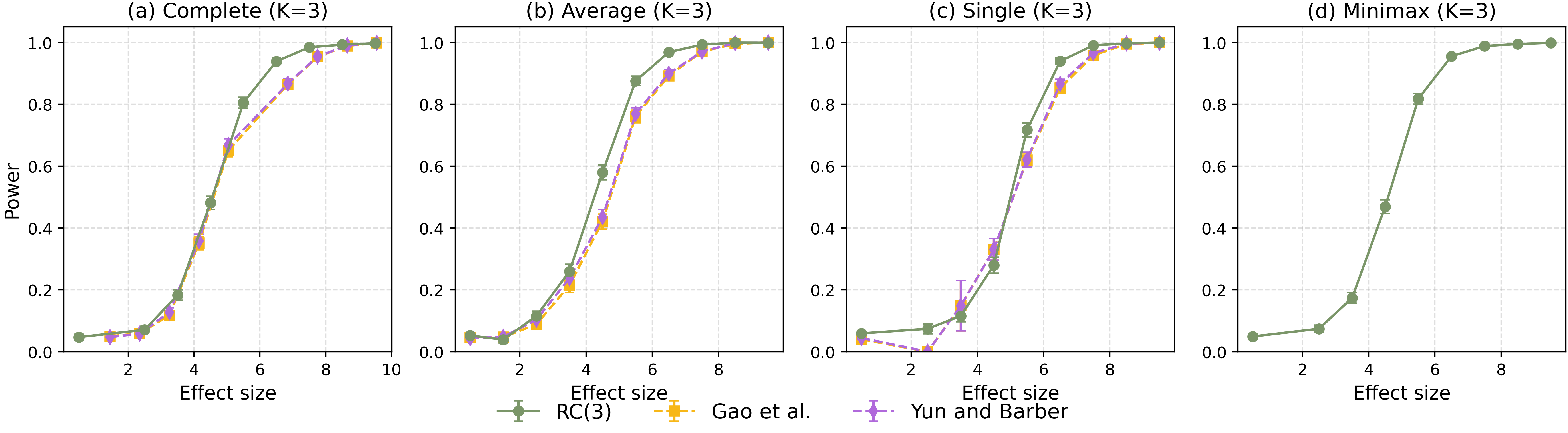}
    \caption{Empirical power curves as a function of effect size for the proposed randomized method with $\tau^* = 0.10$, compared with the two selective inference approaches under varying choice of linkage functions (complete, single, average and minimax) and true number of clusters $K=3$.}
    \label{fig:power_k3}
\end{figure}

Next, we present the comparison of the proposed p-value validity with existing selective inference methods. Using the same data-generating model as in Section \ref{sec:validitypower} with $n=30$, $p=10$, $\sigma =1$, we perform our clustering algorithm until the data are partitioned into $K=2$ and $K=3$ clusters, respectively,
under four linkage criteria: complete, average, single, and minimax. 
For each setting, we perform $2000$ repetitions and evaluate the ECDF of the resulting p-values to assess validity under the null.
The software implementations of the methods proposed by \citep{gao2024selective} and \citep{yun2023selective} do not support testing clusters obtained using the minimax linkage criterion.
Figure \ref{fig:ecdf_linkages} shows that, across all settings, the ECDFs of the p-values from RC(3) align closely with the reference uniform distribution, indicating validity.
As expected, the results are similar to the two existing selective inferential methods whenever these approaches are available for a given linkage criterion.

\begin{figure}[h]
    \centering
    \includegraphics[width=\linewidth]{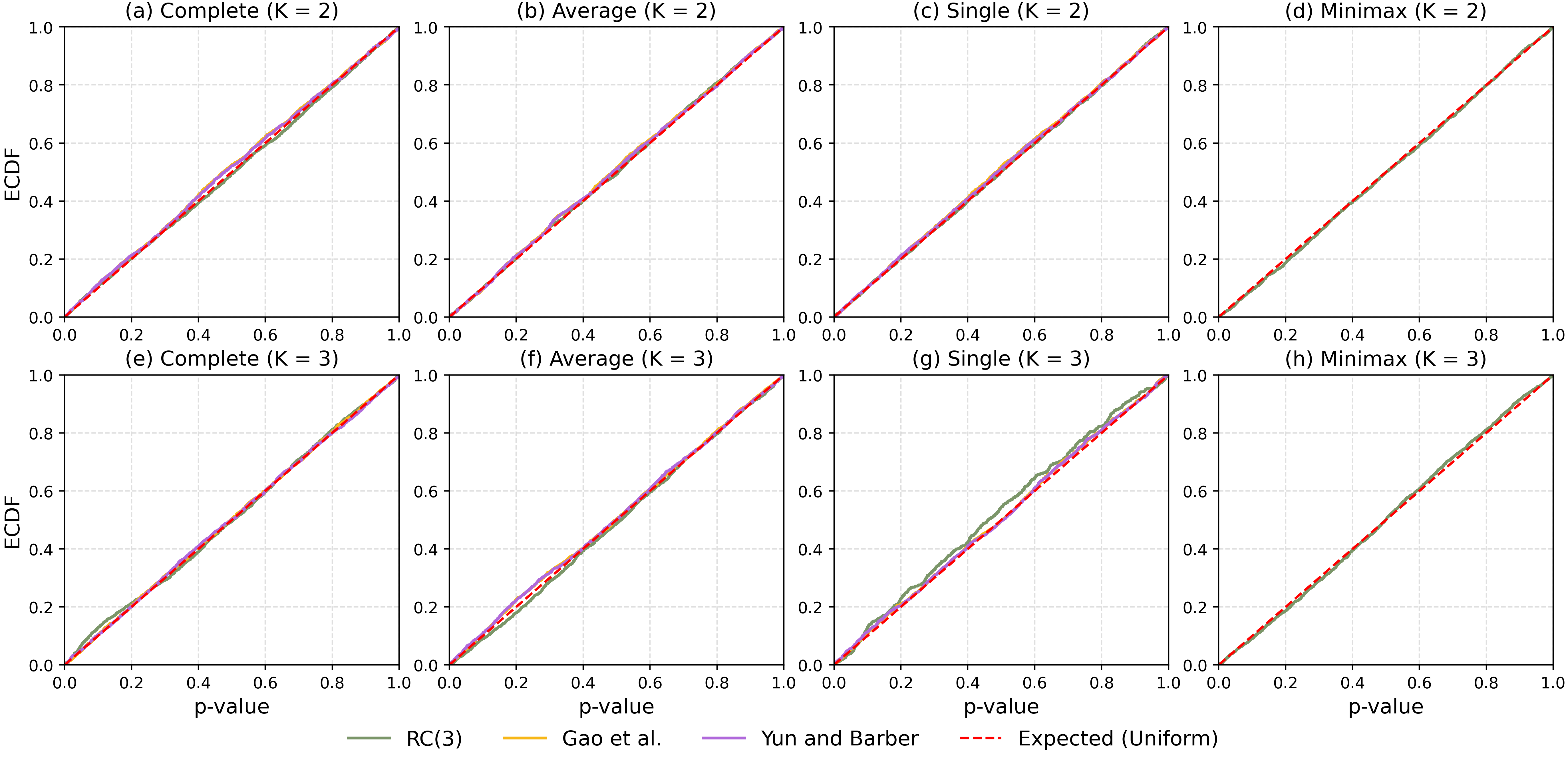}
    \caption{Comparison of the ECDFs of p-values obtained under the null hypothesis. The proposed method with $\tau^* = 0.10$ is compared against \citep{gao2024selective} and \citep{yun2023selective}, across different linkage criteria and cluster numbers $K\in\{2,3\}$.}
    \label{fig:ecdf_linkages}
\end{figure}

\subsection{Recovery of the True Number of Clusters Under Varying $K^*$}\label{subsec:varyK_addtional}
We further evaluate the performance of Algorithm \ref{alg:alpha_spending_stop} under varying true numbers of clusters $K^* \in \{1,2,\dots,10\}$. 
For each value of $K^*$, we generate $100$ independent datasets with $n = 200$ observations and $p = 2$ dimensions. 
The $K^*$ cluster centers are placed evenly around a circle of radius $\delta = 6$, and each data point is sampled from a spherical Gaussian distribution centered at its corresponding cluster. 
Cluster sizes are distributed as evenly as possible across groups. 
We use the same procedure as in the experiments in Section~\ref{sec:chooseK_experiments} to generate the $\alpha$-sequence, and set $n_{\min} = 10$ and $n^* = 40$.

\begin{figure}[h]
    \centering
    \includegraphics[width=\linewidth]{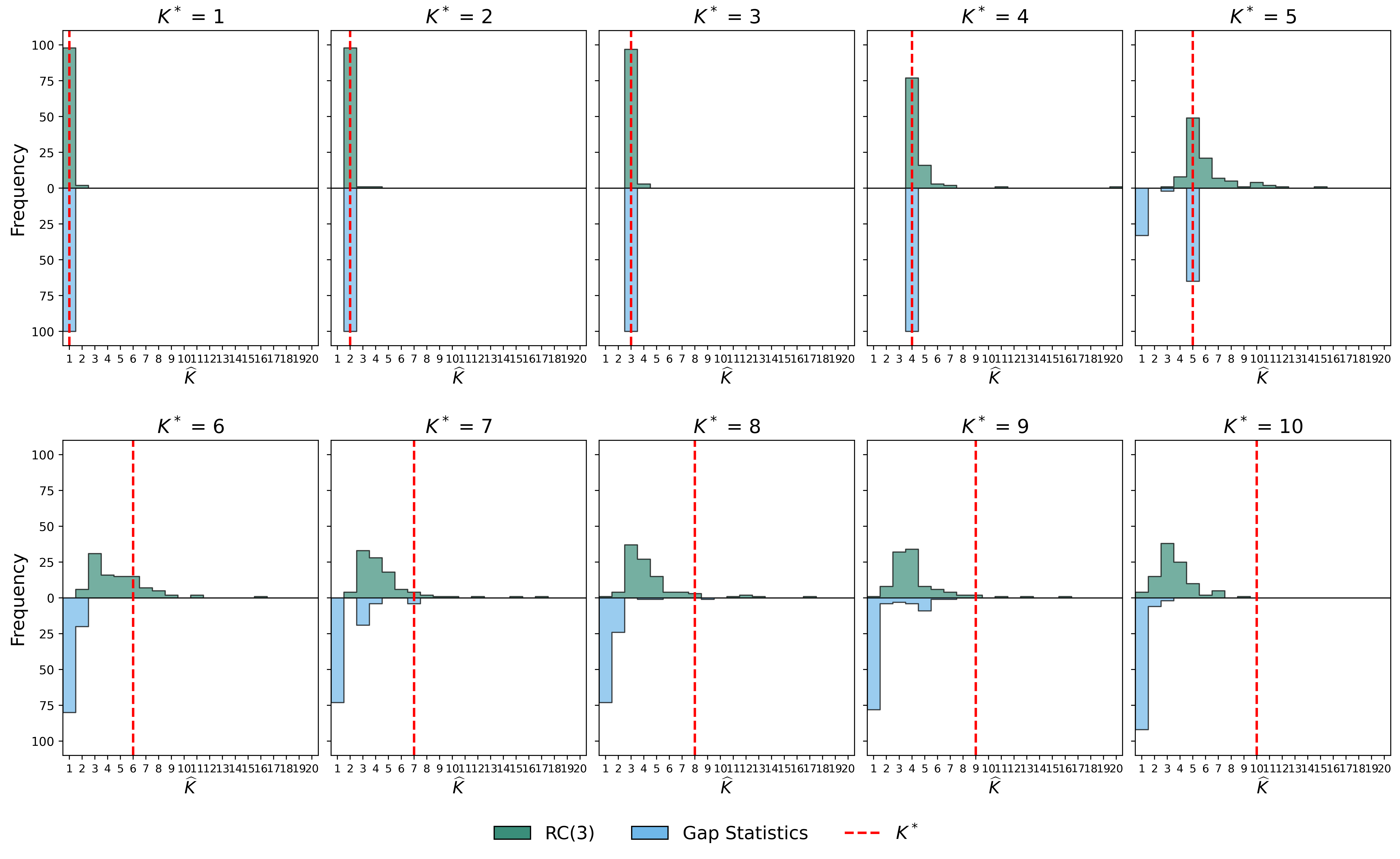}
    \caption{Paired histograms of the estimated number of clusters $\widehat{K}$ under the proposed method and the gap statistics for varying true number of clusters $K^*$.}
    \label{fig:hist_varyingK}
\end{figure}

Figure \ref{fig:hist_varyingK} shows that both the proposed method and the gap statistic recover the true number of clusters with high stability when $K^*$ is small. 
However, as $K^*$ increases, the estimation problem becomes more challenging, and both methods tend to underestimate the true number of clusters. 
Despite this tendency, the proposed procedure often selects more than $\widehat{K} = 1$ cluster, whereas the gap statistic frequently collapses to a single cluster solution. 
In other words, the proposed method reveals substantially more clustering structure in the data than the gap statistic.

\end{document}